\documentclass[aps,twocolumn,superscriptaddress]{revtex4-2}

\usepackage{graphicx}
\usepackage{hyperref}
\hypersetup{
    colorlinks=true,
    linkcolor=magenta,
    citecolor=blue,
    urlcolor=blue
}
\usepackage[utf8]{inputenc}
\usepackage{amsmath, amssymb, bm, amsthm}
\usepackage{braket}
\usepackage{here}
\usepackage{fancybox}
\usepackage{ascmac}
\usepackage{paralist}
\usepackage{comment}
\usepackage{url}
\usepackage{xcolor}
\usepackage[normalem]{ulem}
\usepackage{fancyhdr}

\newtheorem{theorem}{Theorem}
\newtheorem{lemma}{Lemma}

\theoremstyle{remark}

\begin{document}

\title{Provably efficient experimental classification of quantum data with quantum kernels}

\author{Keitaro Anai}
\affiliation{Department of Applied Physics, School of Engineering, The University of Tokyo,
7-3-1 Hongo, Bunkyo-ku, Tokyo 113-8656, Japan}

\author{Akimoto Nakayama}
\affiliation{Graduate School of Engineering Science, The University of Osaka, 1-3 Machikaneyama, Toyonaka, Osaka 560-8531, Japan}

\author{Shintaro Ooka}
\affiliation{Department of Applied Physics, School of Engineering, The University of Tokyo,
7-3-1 Hongo, Bunkyo-ku, Tokyo 113-8656, Japan}

\author{Hiroko Tomoda}
\affiliation{Department of Applied Physics, School of Engineering, The University of Tokyo,
7-3-1 Hongo, Bunkyo-ku, Tokyo 113-8656, Japan}

\author{Yuto Morohoshi}
\affiliation{Graduate School of Engineering Science, The University of Osaka, 1-3 Machikaneyama, Toyonaka, Osaka 560-8531, Japan}

\author{Yohji Chin}
\affiliation{Department of Applied Physics, School of Engineering, The University of Tokyo,
7-3-1 Hongo, Bunkyo-ku, Tokyo 113-8656, Japan}

\author{Shigehito Miki}
\affiliation{Advanced ICT Research Institute, National Institute of Information and Communications Technology, 588-2 Iwaoka, Nishi-ku, Kobe, Hyogo 651-2492, Japan}

\author{Masahiro Yabuno}
\affiliation{Advanced ICT Research Institute, National Institute of Information and Communications Technology, 588-2 Iwaoka, Nishi-ku, Kobe, Hyogo 651-2492, Japan}

\author{Hirotaka Terai}
\affiliation{Advanced ICT Research Institute, National Institute of Information and Communications Technology, 588-2 Iwaoka, Nishi-ku, Kobe, Hyogo 651-2492, Japan}

\author{Peter van Loock}
\affiliation{Institute of Physics, Johannes-Gutenberg University of Mainz, Staudingerweg 7, 55128 Mainz, Germany}

\author{Kosuke Mitarai}
\affiliation{Graduate School of Engineering Science, The University of Osaka, 1-3 Machikaneyama, Toyonaka, Osaka 560-8531, Japan}
\affiliation{Center for Quantum Information and Quantum Biology, The University of Osaka, 1-2 Machikaneyama, Toyonaka, Osaka 560-0043, Japan}

\author{Shuntaro Takeda}
\email{takeda@ap.t.u-tokyo.ac.jp}
\affiliation{Department of Applied Physics, School of Engineering, The University of Tokyo,
7-3-1 Hongo, Bunkyo-ku, Tokyo 113-8656, Japan}

% up to 200 words
\begin{abstract}
    Analyzing quantum systems is essential for quantum information processing, but can require exponentially many measurements without exploiting the structure of the system or task. Quantum machine learning offers a promising way to reduce this cost by learning from quantum data. Here, we theoretically establish a framework for efficient quantum-state classification and experimentally demonstrate the resulting scheme on a programmable photonic processor. Our analysis links learning efficiency to the structure of quantum data, without relying on a physical description tailored to each task. We derive sufficient conditions for accurate classification with a number of measurements polynomial in system size. We apply this guarantee to a resource problem in continuous-variable (CV) quantum computing: classifying stellar rank, a measure of non-Gaussianity. Experimentally, we implement the classification of single-mode CV quantum states with stellar rank zero and one. This is enabled by developing a programmable CV photonic processor capable of preparing and processing diverse Gaussian and non-Gaussian states and performing virtual projection onto such states, going beyond the Gaussian-only capabilities of most existing photonic platforms. The experiment demonstrates accurate classification despite experimental imperfections. Our results establish a theory-guided and experimentally validated framework for quantum-state classification that can be implemented on quantum devices available today.
\end{abstract}

% Word count:
% abstract: 200/200
% introduction: 575
% quantum kernel method: 822
% application of quantum: 509
% experimental kernel and: 542
% conclusion and outlook: 309
%%%%%%%%%%%%%%%%%%%%%%%%%%%%%%%%%
% sum in the main text: 2757/3000

\date{\today}

\maketitle

\section{Introduction}

Analyzing quantum systems—for example, detecting entanglement, identifying phases of matter, and characterizing quantum devices—is essential for quantum information processing. However, such analyses have limitations. Naive analyses without exploiting the structure of the system or task, such as full quantum-state tomography, can require a number of samples that grows exponentially with system size because of the exponential growth of the Hilbert space dimension~\cite{Anshu2024, Mele2025}. To avoid this cost, scientists have been developing methods tailored to individual tasks, using the relevant physical properties to derive suitable measurement strategies, such as entanglement certification of Einstein-Podolsky-Rosen states~\cite{PhysRevLett.84.2722} and characterizing the device noise of Clifford gates~\cite{PhysRevLett.106.180504}. However, such attempts typically require physical insight specific to each problem, and such efficient analysis is not always available for some tasks, such as determining whether an unstructured mixed state is entangled.

Machine learning may offer a promising way to address these limitations. 
For tasks involving classical physical systems, by learning the relationship between input data and target properties from a given dataset, models can make predictions without requiring explicit insight into the underlying physics.
This capacity has motivated quantum machine learning (QML) approaches that take quantum states themselves as quantum data~\cite{PhysRevLett.126.190505, Wu2023quantumphase, doi:10.1126/science.abn7293, Herrmann2022, Saggio2021, GONG2023906}. However, guaranteeing efficient and accurate predictions without task-specific physical assumptions remains difficult, and many existing approaches still rely on the structure of particular tasks~\cite{Wu2023quantumphase, Hsin2022Provably, Lewis2024, morohoshi2025learningfunctionshamiltonianshamiltonian}.
Although some theoretical studies have developed general, task-independent criteria for avoiding overfitting~\cite{Banchi2021Generalization,Caro2022Generalization,Hur2025Margins}, where good performance on the examples used for learning does not carry over to unseen examples, avoiding overfitting alone does not guarantee accurate predictions.

Here, we propose a learning framework using kernel methods with fidelity for this purpose, theoretically establish its efficiency under certain conditions, and experimentally demonstrate the resulting scheme on a photonic quantum processor (Fig.~\ref{fig:concept}). 
Our theoretical analysis links learning efficiency to the structure of quantum data, without relying on a physical description tailored to each task. 
Intuitively, classification is easier when states with the same label are similar and those with different labels are dissimilar. 
To quantify this, we bring a well-established data-structure measure from classical learning theory into the quantum setting. 
We then derive the sufficient conditions based on this measure for accurate classification with a number of measurements polynomial in system size. Separately, when viewed in terms of errors caused by estimating kernels from finitely many shots, our shot-complexity scaling is more favorable than in previous analyses~\cite{Gentinetta2024QSVMComplexity}.

We use this framework on a physically motivated problem: identifying non-Gaussian resources in continuous-variable (CV) quantum states.
Specifically, we classify stellar rank, a resource measure of non-Gaussianity~\cite{PhysRevLett.124.063605}.
Using the common measure described above, we prove that the state families studied here satisfy our sufficient conditions for learning with resources that scale polynomially with system size.

To demonstrate the effectiveness of this framewok experimentally, we establish a programmable CV photonic processor capable of preparing and processing a variety of Gaussian and non-Gaussian states and performing virtual projection onto such states, going beyond the Gaussian-only capabilities of most existing photonic platforms.~\cite{larsen2021deterministic, yokoyama2025fullstackanalogopticalquantum}.
We implement a minimal nontrivial binary instance on a programmable photonic processor, separating single-mode CV states of stellar rank zero from those of rank one. Although minimal, this instance captures several fundamentally and practically important tasks, including distinguishing Gaussian from non-Gaussian states~\cite{PRXQuantum.2.030204}, detecting Wigner negativity~\cite{Chabaud2021witnessingwigner}, and distinguishing vacuum from single-photon states in discrete-variable photonic quantum information processing~\cite{RevModPhys.79.135}. 
Using stellar-rank classification as a concrete testbed, we experimentally demonstrate that the proposed learning scheme remains effective on current quantum hardware despite experimental imperfections. Our results establish a theory-guided and experimentally validated framework for quantum-state classification, with polynomial resource guarantees under general conditions on the structure of quantum data.

\section{Quantum kernel method for quantum data}\label{sec:quantumkernel}
\begin{figure}[htbp]
    \centering
    \includegraphics[width=1.0\linewidth]{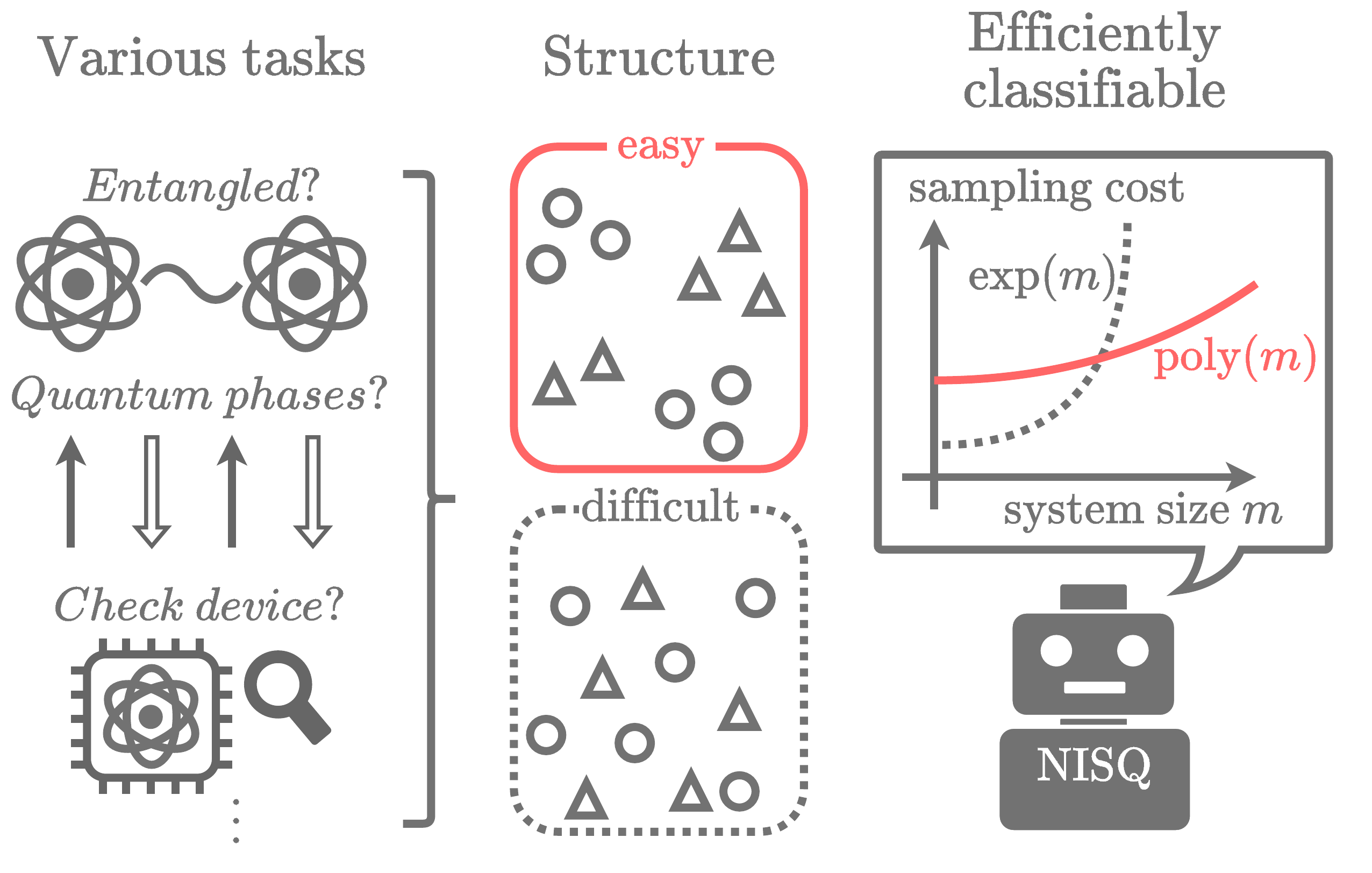}
    \caption{\textbf{Conceptual overview of the proposed quantum-state classification framework.} Our framework is designed to address diverse quantum-state classification tasks motivated by questions in quantum science, such as determining whether a state is entangled, identifying phases of matter, and detecting faulty outputs from quantum devices (left). Across these tasks, we assess the structure of quantum data using a common measure based on a quantum kernel, which quantifies similarity between states (center). For data satisfying our sufficient conditions, the same quantum kernel method enables accurate learning and prediction with numbers of training states and measurement shots that scale polynomially with the system size \(m\), without relying on task-specific physical insight. Our strategy can be implemented even with noisy intermediate-scale quantum (NISQ) devices (right).}
    \label{fig:concept}
\end{figure}

Classifying quantum states by physically meaningful properties is a central task in quantum science.
For instance, we may wish to decide whether a quantum state is entangled or separable, whether the Wigner function of a quantum state has a negative value, or which phase (e.g., paramagnetic, ferromagnetic, or antiferromagnetic) a many-particle quantum system belongs to. 
Such classification tasks are essential for analyzing the outputs of quantum sensors and quantum simulators, for evaluating the resources we generate, and for addressing questions of fundamental physical interest.
However, a straightforward approach to such classification, in which the target state is fully reconstructed and then analyzed classically, generally requires resources that scale exponentially with the system size~\cite{Mele2025}. This cost may be avoided by exploiting problem-specific physical insight, but efficient methods based on such insight are not always available.

To address this limitation, we prove that efficient learning is possible under assumptions that do not rely on problem-specific physical descriptions.
Specifically, we focus on a quantum kernel method~\cite{PhysRevLett.122.040504, havlicek2019supervised} based on fidelities between quantum states.
The quantum kernel method is a supervised learning framework in which a classifier is trained on a classical computer using kernel values estimated by a quantum device.
This framework allows us to draw on the rich learning theory of kernel methods, while the kernel used in this work depends only on fidelities between pairs of states, which can be estimated directly by measurements on a quantum device without state tomography~\cite{PhysRevLett.87.167902, havlicek2019supervised}.
These features align with our aim of providing both theoretical guarantees for efficient learning and an experimental demonstration of the method.

More concretely, we consider the classification tasks described below. 
We are first given a training quantum dataset $\{\ket{\psi_i}\}_{i=1}^{N}$ together with the corresponding labels $\{z_i\}_{i=1}^{N}$. Our goal is to predict the unknown label $z_{\mathrm{test}}$ of a test state $\ket{\psi_{\mathrm{test}}}$. 
The method consists of training and testing steps.
In the training step, a quantum device estimates the kernel values $K(\ket{\psi_i}, \ket{\psi_j})$ for pairs of training states.
The classifier is trained using the kernel values evaluated for all pairs of training quantum data.
In the testing step, the device estimates the kernel values between the test state and the training dataset, and the trained classifier outputs the predicted label $z_{\mathrm{pred}}$.
The detailed algorithm is described in Supplementary Sec.~\ref {app:theorem1-proof}.

One simple quantum kernel that can be estimated experimentally is the fidelity kernel \(K_{\mathrm{fid}}=F\), where \(F(\ket{\psi_i},\ket{\psi_j}):=|\braket{\psi_j|\psi_i}|^2\) for pure states~\cite{havlicek2019supervised}.
To allow for more flexible classification rules while using the same fidelity estimates, we instead consider the quantum RBF kernel~\cite{e25060860},
\begin{align}
K_{\sigma}^{\mathrm{RBF}}(\ket{\psi_i},\ket{\psi_j})
&:=
\exp\!\left[
-\frac{\|\hat{\rho}_j-\hat{\rho}_i\|_F^2}{2\sigma^2}
\right]\\
&=
\exp\!\left[
-\frac{1-F(\ket{\psi_i},\ket{\psi_j})}{\sigma^2}
\right],
\label{eq:RBFkernel}
\end{align}
where $\hat{\rho}_i:=\ket{\psi_i}\bra{\psi_i}$, $\|\cdot\|_F$ denotes the Frobenius norm, and $\sigma>0$ is a tunable hyperparameter. Thus, our scheme is a support-vector-machine classifier using
the quantum RBF kernel constructed from fidelities between quantum data. 
See Supplementary Sec.~\ref{sppl:DiagramOfKernel} for a detailed comparison of the fidelity and RBF kernels. 

For binary classification using this scheme, we establish sufficient conditions for learning from a limited number of training states and measurements (Fig.~\ref{fig:concept}).
Intuitively, if states with the same label are similar and states with different labels are dissimilar, we expect that the states can be classified easily using these similarities.
We formalize this intuition using two ingredients: a kernel that measures similarity between quantum states and representative states that serve as reference points for a class, as follows.
Suppose that one class has $M$ representative states such that every state in this class has a kernel value of at least $\kappa$ with at least one representative, whereas every state in the other class has kernel values of at most $\varepsilon_{\mathrm{sep}}$ with all representatives. We define the margin $\eta=\kappa/M-\varepsilon_{\mathrm{sep}}$, which provides a lower bound on the gap between the two classes in their average kernel values over the representatives. 
Our initial intuition can thus be expressed in terms of this margin: a larger positive margin allows learning with fewer training states and measurements.
The following theorem formalizes this connection.
We note that the conditions stated in the theorem require representative states for only one of the two classes.

\textbf{Theorem 1.}
\label{thm:theorem1_anai}
\textit{
Consider binary classification in which training and test states are drawn independently from a fixed distribution. 
% これらの仮定の上で、さらに〜を仮定する、というニュアンスで書けば誤読を防げそう。
Let $M$ denote the number of representative states for one class, and let $\kappa$ and $\varepsilon_{\mathrm{sep}}$ be the lower and upper bounds on the kernel values defined above, respectively.
Let $\eta=\kappa/M-\varepsilon_{\mathrm{sep}}$ denote the corresponding margin.
Suppose that $\eta$ is bounded below by an inverse polynomial in the system size $m$
(This condition also implies that $M$ is bounded above by a polynomial in $m$).
Suppose also that this class does not collapse onto its representatives.
Then the numbers of training states and measurement shots sufficient for learning and prediction are bounded by polynomials in the system size $m$.
}

The full assumptions and guarantees are given in Supplementary Theorem~\ref{thm:formal-main-polynomial-learning}.
Notably, because the guarantee targets the misclassification probability rather than the output of the trained classifier itself, the total number of measurement shots scales as $O(N^4\log N)$ for training and $O(N\log N)$ for each prediction with $N$ training states, compared with $O(N^{4.67})$ and up to $O(N^2)$, respectively, in a previous analysis~\cite{Gentinetta2024QSVMComplexity} (see Supplementary Sec.~\ref{app:theorem1-proof}).

\section{Application of quantum kernel method: estimation of stellar rank}\label{sec:ApplicationOfQuantum}

\begin{figure}[htbp]
    \centering
    \includegraphics[width=0.95\linewidth]{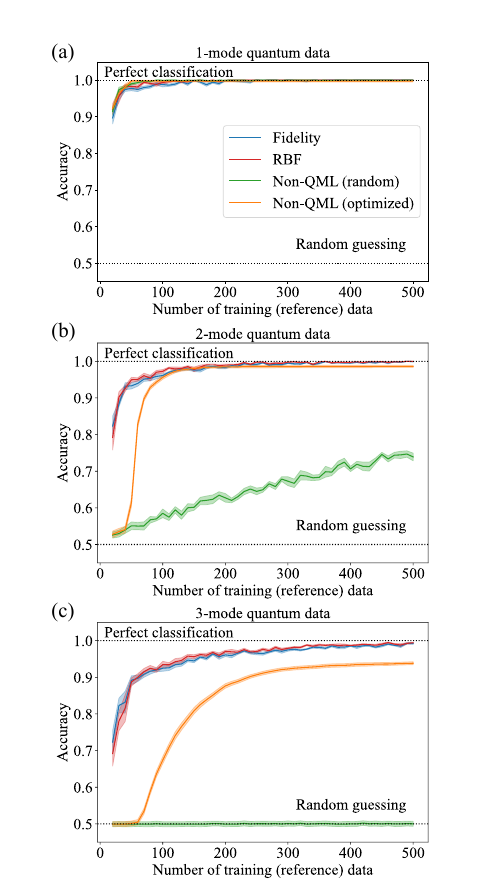}
    \caption{\textbf{Numerical simulation of the mode-number dependence of classification accuracy in stellar-rank-based quantum data classification.} 
    Blue, red, green, and orange lines correspond to results obtained using the fidelity kernel as a linear baseline, quantum RBF kernel, non-QML methods with random-candidate testing, and non-QML methods with optimized-candidate testing, respectively. Here, as the non-QML method, we adopt the approach proposed in Ref.~\cite{PhysRevLett.124.063605}, which is based on a witness called stellar robustness (see Supplementary Sec.~\ref{sppl:StellarRobustnessAnd} for the detailed procedure). 
    For the quantum-kernel methods, the horizontal axis denotes the number of labeled quantum data used to train the classifier, whereas, for the non-QML methods, it denotes the number of reference states of known stellar rank used for threshold-based classification.
    This comparison should be read as a downstream classification benchmark; it does not compare the measurement-shot resources required to obtain the inputs used by each method. The legend, shown only in panel \textbf{a}, applies to all panels.
    The shaded area represents the standard error. 
    \textbf{a–c.} Classification accuracies with one-, two-, and three-mode quantum datasets, respectively.}
    \label{fig:MultiModeAccuracy_main}
\end{figure}

To demonstrate the physical relevance of our quantum kernel method, we apply it to non-Gaussian resource estimation in CV quantum systems.
The non-Gaussian resource is necessary for achieving quantum advantage in CV quantum computations~\cite{PhysRevLett.88.097904, PhysRevResearch.3.033018}. 
In particular, the \textit{stellar rank} is a discrete resource measure that captures the non-Gaussian structure of a CV state. Although previous approaches use witnesses to certify lower bounds on stellar rank~\cite{PRXQuantum.2.020333,fiurasek2022efficient,hct3-6hgn,PhysRevLett.123.043601}, we instead learn a supervised classifier that is provably accurate under the assumptions of Theorem~1.

Here, we formulate a stellar-rank classification task that satisfies the assumptions of Theorem~1.
We consider $m$-mode quantum data of the form $\ket{\psi} = \hat{G}\ket{\bm{n}}$ with its stellar rank $|\bm{n}| = s$, where $\hat{G}$ is an $m$-mode Gaussian operation, $\bm{n} = (n_1, \ldots, n_m) \in \mathbb{Z}_{\geq 0}^m$, $|\bm{n}| = \sum_{i=1}^m n_i$, and $\ket{\bm{n}} = \ket{n_1} \otimes \cdots \otimes \ket{n_m}$ is a tensor product of Fock states~(general definition of the stellar rank is described in Ref.~\cite{PhysRevResearch.3.033018}). The stellar rank $|\bm{n}|$ corresponds to the label $z$ of the classification task.
The precise definitions of the state families considered in the following numerical and experimental analyses are provided in Supplementary Sec.~\ref{sppl:CheckOfAssumptions}, where we verify that they satisfy the assumptions of Theorem~1.
Thus, this task lies within the theorem's finite-shot learnable regime. 

To assess the mode-number dependence of the proposed classifier, we simulate the stellar-rank classification task with QuTiP~\cite{lambert2024qutip5quantumtoolbox}.
We consider one-, two-, and three-mode quantum datasets with stellar rank zero or one, focusing on binary classification as the minimal nontrivial setting. 
The same quantum kernel framework can be extended to higher stellar ranks; see Supplementary Sec.~\ref{sppl:numericalsimulation} for a brief discussion. 
Figure~\ref{fig:MultiModeAccuracy_main} presents the numerical results on how the classification accuracy depends on the number of training $m$-mode quantum data.

We first analyze the performance of our scheme with the quantum RBF kernel. 
For all mode numbers considered here, the accuracy continues to improve with the number of training samples and approaches unity.
This behavior indicates that the RBF kernel retains discriminative information even as the Hilbert-space dimension increases.
Then, we discuss the difference between the quantum RBF kernel and other methods in Fig.~\ref{fig:MultiModeAccuracy_main}.  
In the present simulations, the RBF kernel performs slightly better than the fidelity kernel used as a linear baseline, consistent with the additional expressivity provided by the nonlinear RBF transform.
Moreover, we observe a clear performance difference between the RBF kernel and the non-QML methods. 
The accuracy labeled ``non-QML (random)'' drops sharply as the number of modes increases and becomes close to the random-guessing baseline for three-mode quantum data.
The optimized non-QML baseline partially mitigates this degradation, but it still does not outperform the RBF kernel in the three-mode case.
Thus, within this controlled numerical benchmark, the RBF kernel is the most robust of the tested methods across the one-, two-, and three-mode settings.

\section{Experimental kernel and accuracy}\label{sec:ExperimentalKernelAnd}

Hereafter, we discuss the experimental results. In our experiment, we consider the single-mode quantum dataset $(m=1)$ with stellar rank zero or one as the simplest nontrivial situation of stellar-rank classification.
More formally, we use a labeled set of 180 states \(\{\ket{\psi_i}=\hat G_i\ket{n_i}\}_{i=1}^{180}\), with \(n_i\in\{0,1\}\), and labels \(\{z_i=n_i\}_{i=1}^{180}\)
(see Supplementary Sec.~\ref{sppl:WorkingPrincipleOf} for the detailed definition). 

We experimentally estimate the kernel entries required for this classification using our programmable photonic processor (see \hyperref[sec:methods]{Methods} for details). The processor can generate a variety of quantum states of stellar rank zero (Gaussian) and stellar rank one (non-Gaussian), and perform virtual projection onto such states, going beyond the Gaussian-only capabilities of most existing photonic platforms~\cite{larsen2021deterministic, yokoyama2025fullstackanalogopticalquantum}.

To validate the effectiveness of our quantum kernel method experimentally, we perform two primary assessments: (i) comparison of fidelities between experiment and theory, and (ii) classification accuracy analysis using experimentally and theoretically obtained kernels.

\begin{figure*}[htbp]
    \centering
    \includegraphics[width=0.9\linewidth]{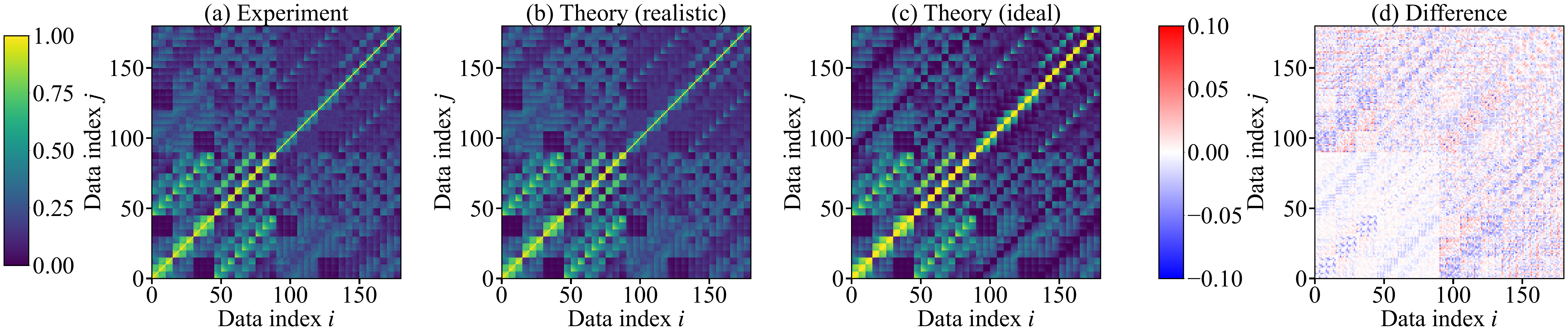}
    \caption{\textbf{Fidelity values $|\braket{\psi_j|\psi_i}|^2$ between quantum data.} The horizontal and vertical axes in Fig.~\ref{fig:kernel_exp} represent the index of the 180 quantum data. In \textbf{a--c}, the colormap represents the fidelity values between quantum data with the corresponding indices; in \textbf{d}, it represents the element-wise difference between \textbf{a} and \textbf{b}. The rows and columns of the matrix are ordered identically. The data are grouped according to their class label, with stellar-rank-zero states placed first and stellar-rank-one states placed second. See Supplementary Sec.~\ref{sppl:WorkingPrincipleOf} for details on the ordering. \textbf{a.} Experimental fidelities. Each element's standard error is approximately 0.004, determined via bootstrap resampling. \textbf{b.} Theoretical fidelities with experimental imperfections. For \textbf{a} and \textbf{b}, the raw fidelities are not positive semi-definite due to the imperfections. Hence, we plot the corrected result (see Supplementary Sec.~\ref{sppl:DiagramOfKernel} for details). \textbf{c.} Theoretical fidelities without experimental imperfections. \textbf{d.} Difference between experimental and theoretical fidelities (Fig.~\ref{fig:kernel_exp}(a) minus Fig.~\ref{fig:kernel_exp}(b)).}
    \label{fig:kernel_exp}
\end{figure*}

First, we evaluate the precision of the measured fidelities between quantum data. Figure~\ref{fig:kernel_exp} shows the experimentally measured fidelities $|\braket{\psi_j|\psi_i}|^2$ alongside the corresponding theoretical predictions. Overall, the experimental fidelities' structure in Fig.~\ref{fig:kernel_exp}(a) reasonably well agrees with the realistic model in Fig.~\ref{fig:kernel_exp}(b), which includes the effect of experimental imperfections such as the optical loss and the noise during the gate operation (see Supplementary Sec.~\ref{sppl:ExperimentalSettingsAnd} for the detailed simulation condition). When compared with the ideal theoretical prediction without including experimental imperfections shown in Fig.~\ref{fig:kernel_exp}(c), the experimental result in Fig.~\ref{fig:kernel_exp}(a) exhibits reduced contrast but nevertheless preserves the overall structure. Figure~\ref{fig:kernel_exp}(d) shows the difference between Figs.~\ref{fig:kernel_exp}(a) and (b). The difference tends to be large when we create stellar-rank-one data $\ket{\psi_i}=\hat{G}_i\ket{1}$, which is attributed to fake counts in photon detection.

\begin{figure}[htbp]
    \centering
    \includegraphics[width=0.9\linewidth]{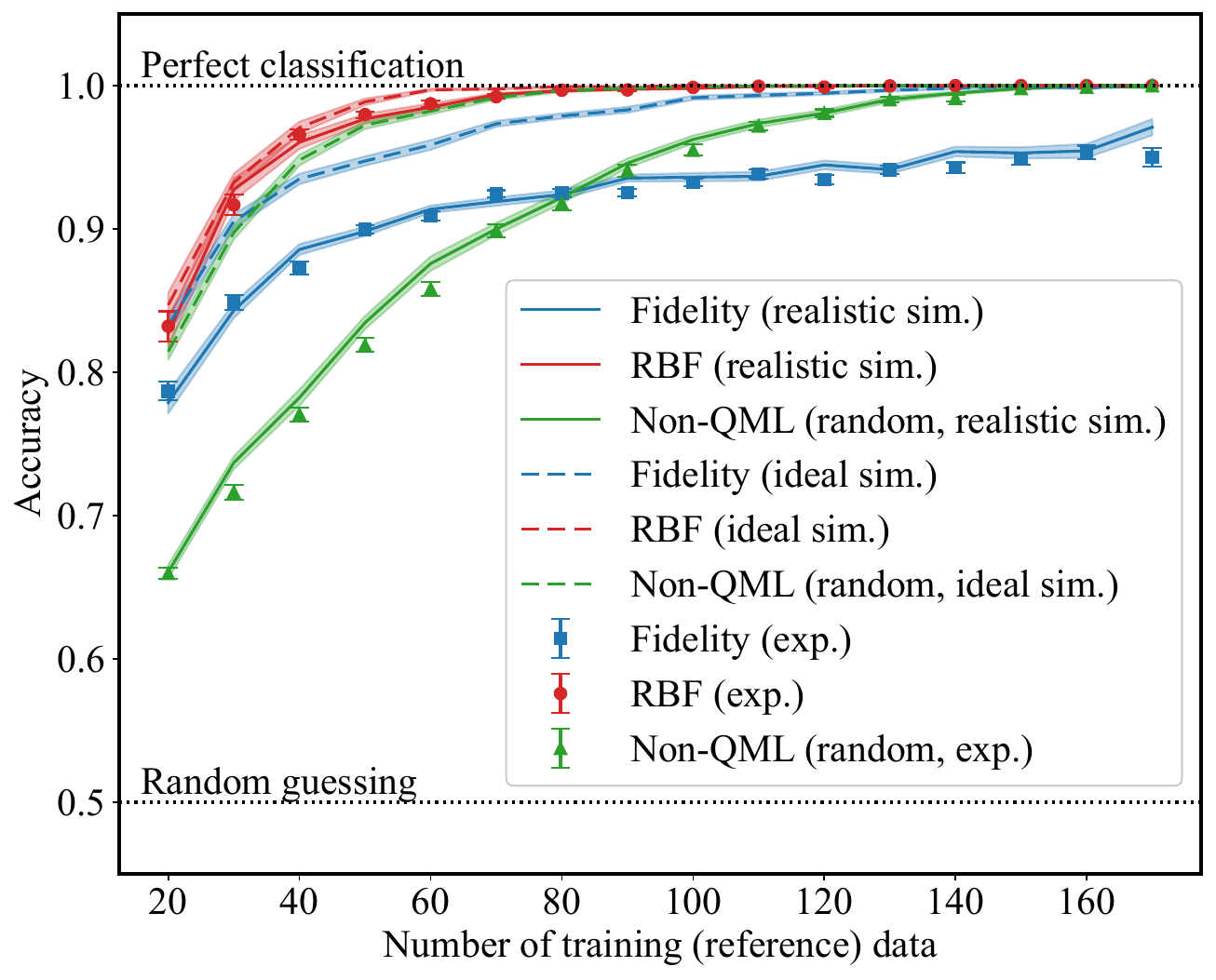}
    \caption{\textbf{Classification accuracy versus training dataset size.} Markers represent experimentally measured accuracies (exp.), solid lines indicate simulated theoretical accuracies including experimental imperfections (realistic sim.), and dashed lines indicate simulated theoretical accuracies without experimental imperfections (ideal sim.). Blue, red, and green correspond to results obtained using the quantum fidelity kernel, quantum RBF kernel, and non-QML methods with random candidates, respectively. The shaded area represents the standard error (see Supplementary Sec.~\ref{sppl:DiagramOfKernel} for details).}
    \label{fig:accuracy_exp}
\end{figure}

Next, we evaluate the classification performance of our quantum kernel method. Figure~\ref{fig:accuracy_exp} compares the classification accuracy as a function of the training-set size for the quantum RBF kernel, the quantum fidelity kernel, and the non-QML baseline. 
The experimental accuracy of the RBF classifier approaches unity as the number of training data points increases and agrees with the realistic theoretical simulation, demonstrating successful experimental learning.
In the presence of experimental imperfections, the RBF classifier shows little degradation relative to the ideal simulation, whereas the fidelity-kernel and non-QML baselines exhibit a more pronounced loss of accuracy.
The comparatively weak degradation of the RBF classifier may reflect the localized, nonlinear weighting induced by the RBF transform: the RBF similarity to each training state decreases exponentially with the fidelity-derived distance $1-F$, with $\sigma$ setting the locality scale, whereas the fidelity kernel yields a global affine decision function in density-operator space.
Within the state family and experimental imperfections considered here, the RBF kernel is therefore the most robust of the tested classifiers, enabling successful learning on the noisy photonic processor.
This result indicates that our scheme is compatible with today's noisy intermediate-scale quantum devices and paves the way toward the quantum advantage with such devices.

\section{Conclusion and outlook}

In this work, we have investigated how the quantum kernel method can enhance our ability to retrieve features of quantum states. We identify sufficient conditions on the structure of quantum data for accurate classification with a number of measurements polynomial in system size. These conditions do not rely on task-specific physical insight or models, providing a common criterion for efficient learning across diverse quantum-state classification tasks. Moreover, we have implemented its proof-of-principle experiment on our CV programmable photonic processor, focusing on the task of classifying quantum states according to their stellar rank—a resource measure of non-Gaussianity for CV quantum computations. This is enabled by developing a programmable CV photonic processor capable of preparing and processing a variety of Gaussian and non-Gaussian states and performing virtual projection onto such states, going beyond the Gaussian-only capabilities of most existing photonic platforms.~\cite{larsen2021deterministic, yokoyama2025fullstackanalogopticalquantum}. Our experimental implementation not only illustrates that our scheme can be applied to the practical resource-estimation task but also highlights that our scheme can be implemented with today's noisy intermediate-scale quantum platforms.

An important future direction will be to demonstrate the sampling advantage experimentally with large-scale quantum devices. For now, we lack an experimental quantum advantage because we deal with small-scale and simple quantum states. However, the recent photonic platform can handle one hundred input qumodes~\cite{yokoyama2025fullstackanalogopticalquantum} and is suitable for manipulating large-scale quantum states, which is essential to achieve experimental quantum advantage. Another future direction will be to identify other practical learning problems where our scheme can be applied. Although we experimentally focus on the specific task of estimating the resource of CV quantum computation, our scheme can retrieve general features of quantum states as long as the quantum states satisfy the assumptions of Theorem~1. Our results establish a theory-guided and experimentally validated route to learning physical resources of quantum states from a finite set of fidelity measurements.

\section*{Acknowledgments}
This work was partly supported by the Japan Science and Technology Agency (JST) Grants No. JPMJCR25I4, No. JPMJFR223R, No. JPMJMS2064, No. JPMJMS256I, and No. JPMJPF2221, and the Japan Society for the Promotion of Science (JSPS) KAKENHI Grants No. 23H01102, No. 23K17300, No. 24KJ0934, and No. 24KJ0726. K.A., S.O., and H.T. acknowledge financial support from The Forefront Physics and Mathematics Program to Drive Transformation (FoPM), WINGS Program, the University of Tokyo. P.v.L. acknowledges support from the BMBF/BMFTR via PhotonQ, QR.N, QuKuK, QuaPhySI, and via the EU project CLUSTEC (grant agreement no. 101080173).

\bibliographystyle{naturemag}
\bibliography{mainrefs}

\clearpage

\section{Methods}\label{sec:methods}

\begin{figure*}[htbp]
    \centering
    \includegraphics[width=0.9\linewidth]{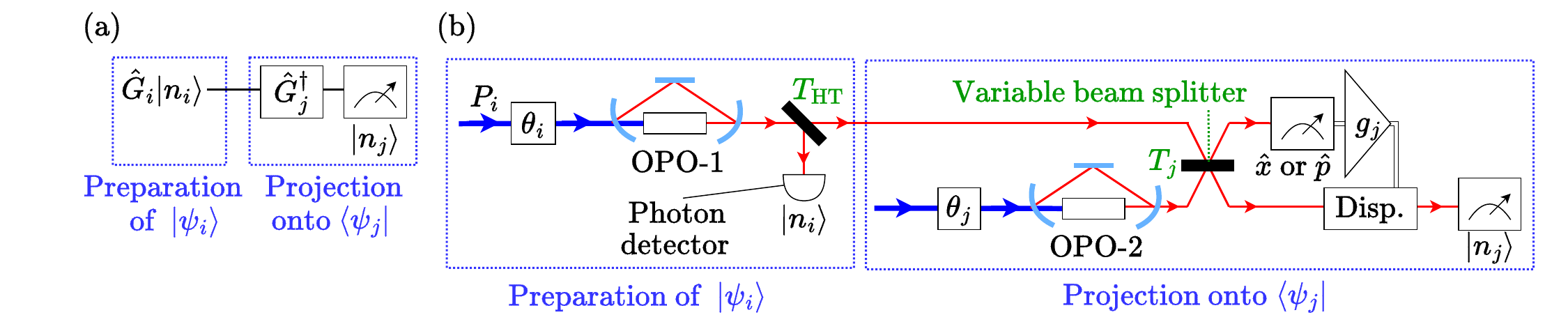}
    % \caption{\textbf{Experimental setup.} 
    \caption{\textbf{Experimental measurement of quantum kernel entries.}
    \textbf{a,} Conceptual quantum circuit for quantum kernel measurement. 
    Quantum data $\ket{\psi_i}$ are generated and the fidelities $|\braket{\psi_j|\psi_i}|^2$ of quantum data are measured by projection onto quantum data $\ket{\psi_j}$. 
    \textbf{b,} Detailed quantum circuit implementation. Squeezed vacuum states or squeezed single-photon states are generated, depending on the measurement result of the photon detector. For the projection onto $\bra{\psi_j}=\bra{n_j}\hat{G}^\dagger_j$, we implement the measurement-induced operation $\hat{G}^\dagger_j$ with an ancillary state from OPO-2, and the projection onto $\bra{n_j}$. Disp., displacement operation; OPO, optical parametric oscillator.}
    \label{fig:setup}
\end{figure*}

Here, we describe the concept of our experimental setup in Fig.~\ref{fig:setup}(a). To implement the quantum kernel method described in Sec.~\ref{sec:quantumkernel}, we experimentally estimate the fidelities \(F(\ket{\psi_i},\ket{\psi_j})=|\braket{\psi_j|\psi_i}|^2\) that define the kernel entries.
Two standard primitives for estimating such fidelities are the SWAP test~\cite{PhysRevLett.87.167902} and the compute-uncompute method~\cite{havlicek2019supervised}. In our experiment, we adopt the latter one for experimental simplicity. In this method, we first prepare quantum data $\ket{\psi_i}$. 
Subsequently, we perform a projective measurement onto another quantum data $\ket{\psi_j}=\hat{G}_j\ket{n_j}$ by applying the inverse unitary operation $\hat{G}_{j}^\dagger$ and measuring the output on the basis of $\{\ket{0},\ket{1}\}$. 
We note that this strategy is available when we know, at least, the classical representation of the training quantum data and can thus implement $\hat{G}_{j}^\dagger$. 
Otherwise, other methods such as the SWAP test can be employed.

We now describe how this conceptual procedure is implemented in our programmable photonic processor in Fig.~\ref{fig:setup}(b). This processor can programmably implement the preparation of arbitrary quantum data $\ket{\psi_i}=\hat{G_i}\ket{n_i}$ as well as the projection onto $\ket{\psi_j}=\hat{G_j}\ket{n_j}$, thereby enabling the estimation of fidelities $|\braket{\psi_j|\psi_i}|^2$ between arbitrary pairs of quantum data. For the preparation of $\ket{\psi_i}$, the quantum data with stellar rank zero (one) are generated if photon detection on the weakly reflected port of the high-transmissivity mirror $T_{\mathrm{HT}}$ does not occur (occurs)~\cite{PhysRevA.55.3184}. The parameters in $\hat{G}_i$ are changed by controlling the power $P_i$ and the phase $\theta_i$ of the pump beam for optical parametric oscillator-1 (OPO-1). For the projection onto $\bra{\psi_j}=\bra{n_j}\hat{G}_j^\dagger$, we first implement arbitrary single-mode Gaussian operations $\hat{G}_{j}^\dagger$ with the measurement-induced operation~\cite{PhysRevA.71.042308}. The parameters of $\hat{G}_{j}^\dagger$ are changed by controlling the phase $\theta_j$ of the pump beam for OPO-2, transmissivity of the beam splitter $T_j$, homodyne measurement basis $\hat{x}$ or $\hat{p}$, and the feedforward gain $g_j$. The homodyne-measurement outputs are post-processed to reflect the displacement operations in $\hat{G}_i$ and $\hat{G}_{j}^\dagger$. Then, the fidelity $|\braket{\psi_j|\psi_i}|^2$ is estimated through the $n_j$-photon component of the output state, which is calculated from the phase-averaged quadrature distribution obtained by the homodyne detection~\cite{PhysRevA.52.R924}. This detection scheme estimates kernel values while avoiding quantum state tomography. See Supplementary Secs.~\ref{sppl:WorkingPrincipleOf} and~\ref{sppl:ExperimentalSettingsAnd} for more details of the experimental setup and parameters.

\clearpage

%% Supplementary Material 開始
\onecolumngrid
\appendix
\pagestyle{plain}
\renewcommand\thesection{\Alph{section}}
\renewcommand{\thefigure}{\thesection\arabic{figure}}
\renewcommand{\thetheorem}{S\arabic{theorem}}
\setcounter{figure}{0} 

\begin{center}
    \textbf{\Large Supplementary Information: Provably efficient experimental classification of quantum data with quantum kernels}
\end{center}
\vspace{1em}

\tableofcontents
\vspace{1em}

\section{ALGORITHM OF QUANTUM KERNEL METHOD}\label{sppl:DiagramOfKernel}

Here, we explain the detailed algorithm of the quantum kernel method introduced in Sec.~\ref{sec:quantumkernel}. Given a training quantum dataset $\{\ket{\psi_i}\}_{i=1}^{N}$ and the corresponding binary labels $\{z_i\}_{i=1}^{N}$, our goal is to predict the unknown label $z_{\mathrm{test}}$ of a test state $\ket{\psi_{\mathrm{test}}}$. We note that once we can classify the binary labels, we can also implement the multiclass classification by repeatedly implementing the binary classification~\cite{991427}. 

Through the whole procedure of binary classification, the essential component is the quantum kernel $K(\ket{\psi_i}, \ket{\psi_j})$, which quantifies the similarity between two quantum data. In this work, we consider two representative kernels: the quantum fidelity kernel $K_{\mathrm{fid}}(\ket{\psi_i},\ket{\psi_j})$ and the quantum radial basis function (RBF) kernel $K_{\sigma}^{\mathrm{RBF}}(\ket{\psi_i},\ket{\psi_j})$, defined respectively by
\begin{align}
    K_{\mathrm{fid}}(\ket{\psi_i},\ket{\psi_j})
    &=
    |\braket{\psi_j|\psi_i}|^2,
    \\
    K_{\sigma}^{\mathrm{RBF}}(\ket{\psi_i},\ket{\psi_j})
    &=
    \exp\left[
        -\frac{1-|\braket{\psi_j|\psi_i}|^2}{\sigma^2}
    \right],
    \label{eq:RBF_sppl}
\end{align}
where $\sigma$ is a real tunable hyperparameter. While many quantum kernel methods for classical data use the fidelity kernel \(K_{\mathrm{fid}}(\ket{\psi_i},\ket{\psi_j})\) between encoded quantum states~\cite{PhysRevLett.122.040504,havlicek2019supervised, PhysRevA.110.022404}, for quantum data, this fidelity kernel has limited expressivity: it is the linear kernel \(K_{\mathrm{fid}}(\ket{\psi_i},\ket{\psi_j})=\mathrm{Tr}(\hat{\rho}_i\hat{\rho}_j)\) on density-operator features \(\hat{\rho}_i=\ket{\psi_i}\bra{\psi_i}\) in Hilbert-Schmidt space.
Thus, as in a classical support vector machine (SVM) with a linear kernel, the classifier is restricted to affine decision functions in the chosen feature space~\cite{Burges1998}. We also note that a kernel similar to our RBF kernel is proposed in Ref.~\cite{Huang2021}, but it only analyzes the reduced density matrix to retrieve local information. 
Other kernel-based approaches for quantum data have also been explored, either without a finite-shot learning guarantee of the type established here~\cite{PhysRevResearch.6.043234, PhysRevResearch.5.043217}, or with a focus on specific phase-learning settings~\cite{Wu2023quantumphase}.

% The binary classification is conducted in two distinct steps, training and testing, as detailed below. In the training step, we first measure the kernel values between all pairs of the training quantum data. Then, using the calculated kernel values, a binary classifier is constructed by solving the following convex quadratic optimization problem (support vector machine formulation):
% \begin{align}
%     &\min_{\bm{\alpha}} \frac{1}{2}\bm{\alpha}^TQ\bm{\alpha}-\sum_{i=1}^{N}\alpha_i \nonumber\\
%     &\text{subject to}\quad \bm{z}^T\bm{\alpha}=0,\quad 0\leq\alpha_i\leq C\quad (i=1,\dots,N),
%     \label{Eq:TrainingStep}
% \end{align}
% where $Q$ is an $N\times N$ positive semi-definite matrix given by $Q_{i,j}=z_i z_j K(\ket{\psi_i},\ket{\psi_j})$, $\bm{z}=(z_1,\dots,z_N)^T$, $\bm{\alpha}=(\alpha_1,\dots,\alpha_N)^T$, and $C>0$ is a regularization parameter.

% Once the classifier is trained, the label $z_{\mathrm{test}}$ of an unseen quantum data $\ket{\psi_{\mathrm{test}}}$ is predicted {(we denote the predicted label by $z_{\mathrm{pred}}$)} by evaluating the kernel values between the test data and the training dataset, followed by the classification rule
% \begin{equation}
%     z_{\mathrm{pred}}=\mathrm{sgn}\left(\sum_{i=1}^{N}z_i\alpha_iK(\ket{\psi_i},\ket{\psi_{\mathrm{test}}})+b\right),
% \end{equation}
% where $\mathrm{sgn}(\cdot)$ denotes the sign function, and the bias term $b\in\mathbb{R}$ is derived straightforwardly from the training solution $(\bm{\alpha},\bm{z},K(\cdot,\cdot))$~\cite{10.1145/1961189.1961199}.

Hereafter, we explain the detailed settings and parameters of machine learning in our simulation and experiment in Secs.~\ref{sec:ApplicationOfQuantum}~and~\ref{sec:ExperimentalKernelAnd}. To implement classification tasks based on our quantum kernel data, we have to measure the fidelities of quantum data and calculate the kernel. Since the experimental kernel is inevitably affected by experimental imperfections and is no longer positive semi-definite, we correct these kernels by finding the closest positive semi-definite kernel with the CVXPY package~\cite{diamond2016cvxpy, agrawal2018rewriting}, as is explained in Ref.~\cite{havlicek2019supervised}. Then, we construct an SVM-based classification model using the scikit-learn library in Python~\cite{10.1145/1961189.1961199}. The classification model has two tunable hyperparameters: the width parameter $\sigma$ of the RBF kernel and the regularization parameter, which controls the balance between fitting the training data and limiting the complexity of the classifier. The precise description of the SVM, including its training objective, the definition of the regularization parameter, and the prediction rule, is given in Supplementary Sec.~\ref{app:theorem1-proof}. We optimize these hyperparameters by cross-validation using a grid search over predefined ranges.

To evaluate how classification performance varies with the number of training samples, we plot the classification accuracy as a function of training set size in Fig.~\ref{fig:MultiModeAccuracy_main}. For each point on the curve, we begin by generating the training samples whose size ranges from 20 to 500 in steps of 10. For each training size, we tune the hyperparameters by cross-validation on the training samples and evaluate the accuracy using an independently generated test quantum dataset of 500 samples. This entire process is repeated for 10 independently generated datasets, and the resulting $10$ accuracy values are used to compute the mean accuracy and standard error for each training size.

A similar procedure is used in Fig.~\ref{fig:accuracy_exp}, where we use 180 quantum samples and vary the number of training samples from 20 to 170. For each training size, we repeat the train-test split, hyperparameter tuning on the training sample, and evaluation 100 times to compute the average accuracy and standard error.

\section{FORMAL LEARNING GUARANTEE AND PROOF OF THEOREM~1}
\label{app:theorem1-proof}

\begin{figure}
  \centering
  \includegraphics[width=\linewidth]{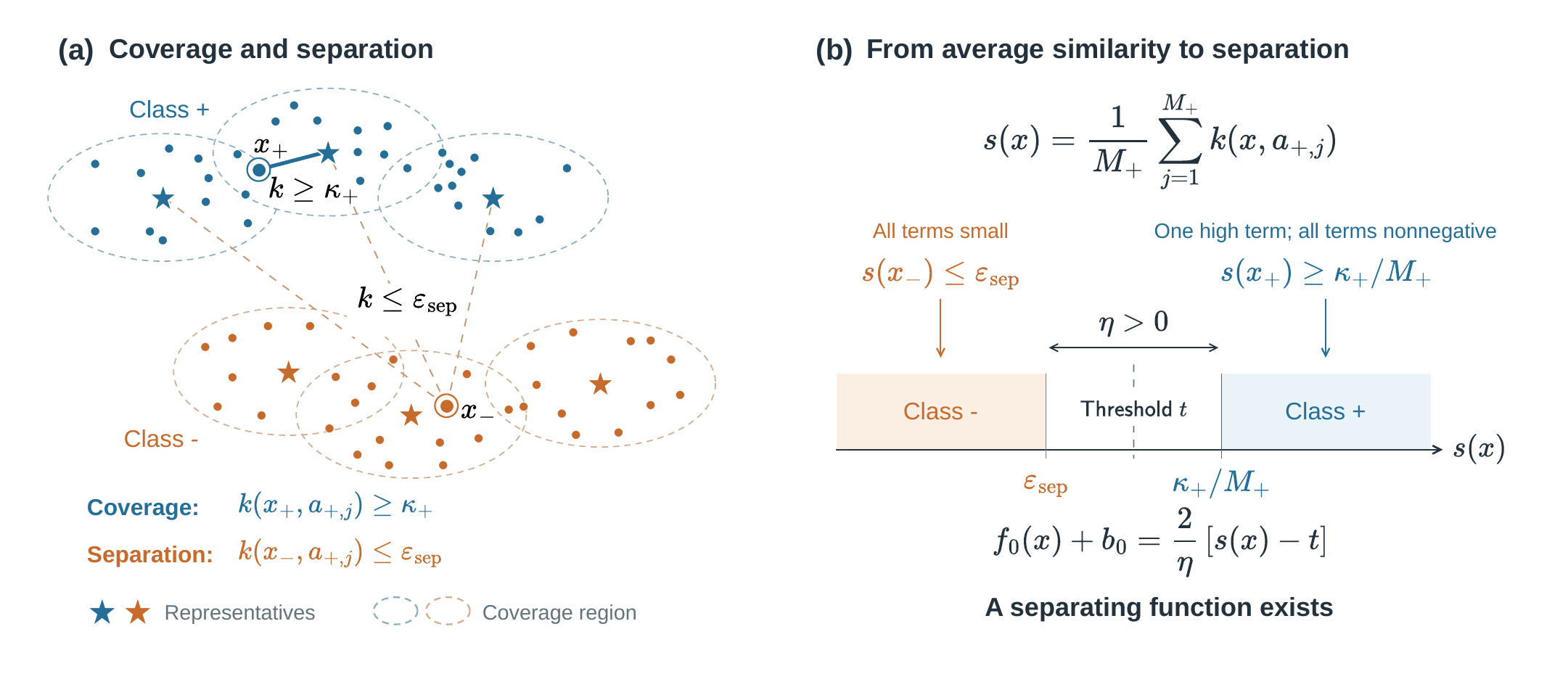}
  \caption{\textbf{From coverage and separation to a unit-margin separator.}
(a) Blue and orange dots denote data points from the two classes, and stars denote their representatives. Dashed ellipses schematically depict coverage regions $\{x\in\mathcal X:k(x,a_{y,j})\ge\kappa_y\}$. 
Every $x\in\mathcal S_y$ lies in at least one region associated with its own class, whereas $k(x,a_{-y,j})\le\varepsilon_{\mathrm{sep}}$ for every opposite-class representative. 
The solid and dashed connections illustrate coverage and separation, respectively.
(b) In this illustration, both classes are evaluated using the same set of blue representatives. 
The score $s(x)=M_+^{-1}\sum_{j=1}^{M_+}k(x,a_{+,j})$ satisfies $s(x_-)\le\varepsilon_{\mathrm{sep}}$ by separation and $s(x_+)\ge\kappa_+/M_+$ by coverage and kernel nonnegativity. 
These bounds leave a positive gap $\eta=\kappa_+/M_+-\varepsilon_{\mathrm{sep}}$. 
The threshold $t$ at the midpoint of this gap separates the classes. 
Centering the score at $t$ and rescaling it yields a separating function $f_0(x)+b_0$ with unit margin almost surely.}
  \label{fig:coverage-separation}
\end{figure}

This section proves Theorem~1.
Theorem~1 extends the analysis of Ref.~\cite{Gentinetta2024QSVMComplexity} in the sense that it also accounts for generalization.
Reference~\cite{Gentinetta2024QSVMComplexity} bounds the number of measurement shots needed to solve the SVM to a given accuracy relative to the solution with the exact kernel; Theorem~1 combines such shot-noise bounds with a generalization bound to guarantee accurate classification of unseen quantum data.
Under the coverage and separation conditions, there exists a unit-margin separator with a bounded norm in the associated reproducing kernel Hilbert space (RKHS) (Sec.~\ref{app:margin-certificate}).
For a fixed classification problem and kernel, both this norm bound and the margin $\eta$ are independent of the number of training states.
Therefore, the regularization parameter $\lambda_{\mathrm{svm}}$ can be chosen from the target misclassification probability and the margin $\eta$ alone, independently of the number $N$ of training states.
The total numbers of measurements then scale as $O(N^4\log N)$ for training and $O(N\log N)$ for each prediction.
For comparison, Ref.~\cite{Gentinetta2024QSVMComplexity} analyzes the accuracy of solving the SVM for a fixed regularization constant multiplying an unnormalized sum of squared hinge losses, corresponding to $\lambda_{\mathrm{svm}}\propto1/N$ in our notation, and derives $O(N^{4.67})$ for training; for each prediction, it estimates $O(N_{\mathrm{SV}}^2)$ in terms of the number $N_{\mathrm{SV}}\le N$ of support vectors, i.e., up to $O(N^2)$, without analyzing prediction further.
Its formal statement is given in Theorem~\ref{thm:formal-main-polynomial-learning}.
The lemmas below establish the intermediate results used in its proof.
Figure~\ref{fig:coverage-separation} illustrates the distributional assumptions of Theorem~\ref{thm:formal-main-polynomial-learning} and the construction of a separating function under these assumptions.

\subsection{Training with finite measurements}
\label{app:training-finite-measurements}

This subsection bounds the expected hinge loss of the score obtained by training with finite measurements.
Section~\ref{app:generic-approx-gram-svm} derives the bound for a general approximate Gram matrix, Sec.~\ref{app:training-finite-shot-gram} controls the error of the Gram matrix estimated from finite measurements, and Sec.~\ref{app:margin-certificate} constructs the unit-margin separator required in Sec.~\ref{app:generic-approx-gram-svm}.

\subsubsection{Expected hinge loss with an approximate Gram matrix}
\label{app:generic-approx-gram-svm}

This subsection derives a bound on the expected hinge loss of an SVM trained with an approximate Gram matrix.

Let $\mathcal X$ be a measurable space and $P$ a probability distribution on $\mathcal X\times\{-1,+1\}$.
For an integer $N\ge1$, let $D_N:=((X_i,Y_i))_{i=1}^N$ be an i.i.d. training sample of size $N$ from $P$.
Let $k:\mathcal X\times\mathcal X\to\mathbb R$ be a measurable symmetric positive semidefinite (PSD) kernel with a separable RKHS $\mathcal H_k$, and assume
\begin{equation*}
    k(x,x)\le1
    \qquad
    \text{for all }x\in\mathcal X.
\end{equation*}
For a measurable score $g:\mathcal X\to\mathbb R$, define its expected hinge loss by
\begin{equation*}
    R_{\mathrm{hinge}}(g)
    :=
    \mathbb E_{(X,Y)\sim P}
    \left[
        \max\{1-Yg(X),0\}
    \right].
\end{equation*}
Fix $\lambda_{\mathrm{svm}}>0$. The SVM objective is
\begin{equation*}
    \mathcal J(f,b)
    :=
    \frac{\lambda_{\mathrm{svm}}}{2}\|f\|_{\mathcal H_k}^2
    +
    \frac1N\sum_{i=1}^N
    \max\{1-Y_i(f(X_i)+b),0\},
    \qquad
    f\in\mathcal H_k,\ b\in\mathbb R.
\end{equation*}
Define the exact Gram matrix on $D_N$ by
\begin{equation*}
    G_{ij}:=k(X_i,X_j).
\end{equation*}
For $f=\sum_{i=1}^N\theta_ik(\,\cdot\,,X_i)$ with $\theta\in\mathbb R^N$, we have
\begin{equation}
    \mathcal J(f,b)
    =
    \frac{\lambda_{\mathrm{svm}}}{2}\theta^\top G\theta
    +
    \frac1N\sum_{i=1}^N
    \max\left\{
            1-Y_i\bigl((G\theta)_i+b\bigr),0
    \right\}
    =:
    J(\theta,b;G),
    \label{eq:generic-svm-objectives}
\end{equation}
and likewise $J(\theta,b;M)$ for any symmetric $M\in\mathbb R^{N\times N}$.
Define the set
\begin{equation}
    \Theta_{\lambda_{\mathrm{svm}}}
    :=
    \left\{
        \theta\in\mathbb R^N:
    0\le Y_i\theta_i\le\frac{1}{\lambda_{\mathrm{svm}}N}
        \ \text{for all }i
    \right\}.
    \label{eq:generic-theta-box}
\end{equation}

The following lemma follows from standard representer and uniform RKHS generalization results~\cite{Kanamori2015StatisticalLearningTheory}.

\begin{lemma}[Expected hinge-loss bound under entrywise Gram matrix error]
    \label{lem:generic-approx-gram-svm}
    Fix constants $R\ge0$, $B\ge0$, $\xi\ge0$, and $\delta\in(0,1)$.
    Suppose that there exist $f_0\in\mathcal H_k$ and
    $b_0\in[-B,B]$ with $\|f_0\|_{\mathcal H_k}\le R$ such that
    \begin{equation}
        Y\bigl(f_0(X)+b_0\bigr)\ge1
        \qquad
        \text{almost surely for }(X,Y)\sim P.
        \label{eq:generic-unit-margin-separator}
    \end{equation}
    For a symmetric PSD matrix $G_{\mathrm{imp}}\in\mathbb R^{N\times N}$, let $(\widetilde\theta_N,\widetilde b_N)$ be any measurable minimizer
    of $J(\,\cdot\,;G_{\mathrm{imp}})$ over
    $\Theta_{\lambda_{\mathrm{svm}}}\times[-B,B]$:
    \begin{equation}
        (\widetilde\theta_N,\widetilde b_N)
        \in
        \operatorname*{arg\,min}_{\substack{
            \theta\in\Theta_{\lambda_{\mathrm{svm}}}\\
            |b|\le B
        }}
        J(\theta,b;G_{\mathrm{imp}}).
        \label{eq:generic-svm-minimizer}
    \end{equation}
    Using the minimizer in Eq.~\eqref{eq:generic-svm-minimizer}, define the
    score evaluated with the exact kernel $k$ by
    \begin{equation}
        g_N(x)
        :=
        \sum_{i=1}^N
        \widetilde\theta_{N,i}k(x,X_i)
        +
        \widetilde b_N.
        \label{eq:generic-reference-score}
    \end{equation}
    Set
    \begin{equation}
        R_\xi
        :=
        \left(
            R^2
            +
            \frac{6\xi}{\lambda_{\mathrm{svm}}^2}
        \right)^{1/2},
        \label{eq:generic-output-radius}
    \end{equation}
    and define
    \begin{equation}
        \Delta_N(R_\xi,B,\delta)
        :=
        \frac{2(R_\xi+B)}{\sqrt N}
        +
        (1+R_\xi+B)
        \sqrt{
            \frac{\log(1/\delta)}{2N}
        }.
        \label{eq:generic-deviation}
    \end{equation}
    Then, with probability at least $1-\delta$ over
    the random draw of $D_N$, every symmetric PSD matrix $G_{\mathrm{imp}}$ with $\max_{i,j}|(G_{\mathrm{imp}})_{ij}-G_{ij}|\le\xi$ and every corresponding
    minimizer in
    Eq.~\eqref{eq:generic-svm-minimizer} satisfy
    \begin{equation}
        R_{\mathrm{hinge}}(g_N)
        \le
        \frac{\lambda_{\mathrm{svm}}R^2}{2}
        +
        \frac{3\xi}{\lambda_{\mathrm{svm}}}
        +
        \Delta_N(R_\xi,B,\delta).
        \label{eq:generic-population-risk-bound}
    \end{equation}
\end{lemma}

To prove Lemma~\ref{lem:generic-approx-gram-svm}, we use the following standard results from Ref.~\cite{Kanamori2015StatisticalLearningTheory}, restated in our notation.

\begin{lemma}[Representer theorem~{\cite[Theorem~4.10]{Kanamori2015StatisticalLearningTheory}}]
    \label{lem:representer}
    Let $k$ be a PSD kernel on $\mathcal X$ with RKHS $\mathcal H_k$, and let $(x_1,y_1),\ldots,(x_N,y_N)\in\mathcal X\times\{-1,+1\}$.
    Let $\mathcal L:\mathbb R^N\to\mathbb R$ be any function, which may depend on $(x_i,y_i)_{i=1}^N$, and let $\Psi:[0,\infty)\to\mathbb R$ be nondecreasing.
    Consider minimizing
    \begin{equation*}
        \mathcal L\bigl(f(x_1)+b,\ldots,f(x_N)+b\bigr)
        +
        \Psi\bigl(\|f\|_{\mathcal H_k}^2\bigr)
    \end{equation*}
    over $f\in\mathcal H_k$ and $b\in\mathbb R$.
    Then there exists an optimal solution in which $f$ has the form
    \begin{equation*}
        f=\sum_{i=1}^N\alpha_ik(\,\cdot\,,x_i),
        \qquad
        \alpha\in\mathbb R^N.
    \end{equation*}
\end{lemma}

\begin{lemma}[Dual problem and optimality conditions of the SVM~{\cite[Sec.~5.3.1, pp.~82--83]{Kanamori2015StatisticalLearningTheory}}]
    \label{lem:svm-dual}
    Let $k$, $\mathcal H_k$, and $(x_i,y_i)_{i=1}^N$ be as in Lemma~\ref{lem:representer}, let $C>0$, and set $K_{ij}:=k(x_i,x_j)$.
    By Lemma~\ref{lem:representer}, minimizing
    \begin{equation*}
        C\sum_{i=1}^N\max\{1-y_i(f(x_i)+b),0\}
        +
        \frac12\|f\|_{\mathcal H_k}^2
    \end{equation*}
    over $f\in\mathcal H_k$ and $b\in\mathbb R$ reduces to the problem
    \begin{equation*}
    \begin{aligned}
        &\min_{\alpha,b,\zeta}\ C\sum_{i=1}^N\zeta_i+\frac12\sum_{i,j}\alpha_i\alpha_jK_{ij}\\
        &\text{subject to}\ \ \zeta_i\ge0,\ \ \zeta_i\ge1-y_i\Bigl(\sum_{j=1}^N\alpha_jK_{ji}+b\Bigr),\ \ i=1,\ldots,N,
    \end{aligned}
    \end{equation*}
    over $\alpha\in\mathbb R^N$, $b\in\mathbb R$, and $\zeta\in\mathbb R^N$.
    Since $K$ is symmetric PSD, this is a convex optimization problem.
    Its dual problem is
    \begin{equation*}
    \begin{aligned}
        &\max_{\gamma\in\mathbb R^N}\ -\frac12\sum_{i,j}\gamma_i\gamma_jy_iy_jK_{ij}+\sum_{i=1}^N\gamma_i\\
        &\text{subject to}\ \ \sum_{i=1}^N\gamma_iy_i=0,\ \ 0\le\gamma_i\le C,\ \ i=1,\ldots,N,
    \end{aligned}
    \end{equation*}
    and optimal primal and dual solutions satisfy
    \begin{equation*}
        \sum_{j=1}^N\alpha_jK_{ji}
        =
        \sum_{j=1}^N\gamma_jy_jK_{ji},
        \qquad
        \sum_{i=1}^N\gamma_iy_i=0,
        \qquad
        0\le\gamma_i\le C,
        \qquad
        i=1,\ldots,N.
    \end{equation*}
\end{lemma}

For a set $\mathcal F$ of real-valued functions on a set $\mathcal Z$ and points $S=(z_1,\ldots,z_N)\in\mathcal Z^N$, the empirical Rademacher complexity is
\begin{equation*}
    \widehat{\mathfrak R}_S(\mathcal F)
    :=
    \frac1N
    \mathbb E_\varsigma\Bigl[
        \sup_{g\in\mathcal F}\sum_{i=1}^N\varsigma_ig(z_i)
    \Bigr],
\end{equation*}
where $\varsigma_1,\ldots,\varsigma_N$ are independent and uniformly distributed on $\{-1,+1\}$.
When $S=(Z_1,\ldots,Z_N)$ consists of i.i.d.\ random variables, the expected Rademacher complexity is
$\mathfrak R_N(\mathcal F):=\mathbb E_S[\widehat{\mathfrak R}_S(\mathcal F)]$~\cite[Definitions~2.4, and~2.5]{Kanamori2015StatisticalLearningTheory}.

\begin{lemma}[Properties of Rademacher complexity~{\cite[Theorem~2.6(4),(5)]{Kanamori2015StatisticalLearningTheory}}]
    \label{lem:rademacher-properties}
    Let $\mathcal F,\mathcal F_1,\ldots,\mathcal F_n$ be sets of real-valued functions on $\mathcal Z$, and let $S\in\mathcal Z^N$.
    \begin{enumerate}
        \item[(i)] (Talagrand's lemma)
        If $\phi:\mathbb R\to\mathbb R$ is Lipschitz continuous with Lipschitz constant $L_\phi$, then
        \begin{equation*}
            \widehat{\mathfrak R}_S(\phi\circ\mathcal F)
            \le
            L_\phi\widehat{\mathfrak R}_S(\mathcal F),
        \end{equation*}
        where $\phi\circ\mathcal F:=\{z\mapsto\phi(g(z))\mid g\in\mathcal F\}$.
        \item[(ii)]
        $\widehat{\mathfrak R}_S\bigl(\sum_{l=1}^n\mathcal F_l\bigr)\le\sum_{l=1}^n\widehat{\mathfrak R}_S(\mathcal F_l)$, where
        $\sum_{l=1}^n\mathcal F_l:=\{\sum_{l=1}^ng_l\mid g_l\in\mathcal F_l,\ l=1,\ldots,n\}$.
    \end{enumerate}
    Taking the expectation over $S$ shows that the same inequalities hold for $\mathfrak R_N$.
\end{lemma}

\begin{lemma}[Rademacher complexity of RKHS balls~{\cite[Theorem~4.11 and Corollary~4.12]{Kanamori2015StatisticalLearningTheory}}]
    \label{lem:rkhs-rademacher}
    Let $k$ be a PSD kernel on $\mathcal X$ with RKHS $\mathcal H_k$, let $r>0$, and let $\mathcal F\subset\{f\in\mathcal H_k\mid\|f\|_{\mathcal H_k}\le r\}$.
    For every $S=(x_1,\ldots,x_N)\in\mathcal X^N$,
    \begin{equation*}
        \widehat{\mathfrak R}_S(\mathcal F)
        \le
        \frac rN\Bigl(\sum_{i=1}^Nk(x_i,x_i)\Bigr)^{1/2}.
    \end{equation*}
    If, in addition, $\sup_{x\in\mathcal X}k(x,x)\le\Lambda^2$, then
    \begin{equation*}
        \widehat{\mathfrak R}_S(\mathcal F)\le\frac{r\Lambda}{\sqrt N},
        \qquad
        \mathfrak R_N(\mathcal F)\le\frac{r\Lambda}{\sqrt N}.
    \end{equation*}
\end{lemma}

\begin{lemma}[Uniform law of large numbers~{\cite[Theorem~2.7]{Kanamori2015StatisticalLearningTheory}}]
    \label{lem:ulln}
    Let $\mathcal F\subset\{g:\mathcal Z\to[c_-,c_+]\}$, and let $Z_1,\ldots,Z_N$ be independent random variables with the same distribution as a random variable $Z$.
    Then, for every $\delta\in(0,1)$, with probability at least $1-\delta$,
    \begin{equation*}
        \sup_{g\in\mathcal F}
        \Bigl\{
            \mathbb E[g(Z)]
            -
            \frac1N\sum_{i=1}^Ng(Z_i)
        \Bigr\}
        \le
        2\mathfrak R_N(\mathcal F)
        +
        (c_+-c_-)\sqrt{\frac{\log(1/\delta)}{2N}}.
    \end{equation*}
\end{lemma}

\begin{proof}[Proof of Lemma~\ref{lem:generic-approx-gram-svm}]
    Every $\theta\in\Theta_{\lambda_{\mathrm{svm}}}$ satisfies
    $\|\theta\|_1\le1/\lambda_{\mathrm{svm}}$ by Eq.~\eqref{eq:generic-theta-box}.
    Hence, for every $G_{\mathrm{imp}}$ with $\max_{i,j}|(G_{\mathrm{imp}})_{ij}-G_{ij}|\le\xi$, every such $\theta$
    satisfies
    \begin{equation*}
        \|(G_{\mathrm{imp}}-G)\theta\|_\infty
        \le
        \xi\|\theta\|_1
        \le
        \frac{\xi}{\lambda_{\mathrm{svm}}},
        \qquad
        \frac{\lambda_{\mathrm{svm}}}{2}
        |\theta^\top(G_{\mathrm{imp}}-G)\theta|
        \le
        \frac{\xi}{2\lambda_{\mathrm{svm}}}.
    \end{equation*}
    Since the hinge loss has Lipschitz constant one, these estimates imply,
    uniformly over all $(\theta,b)\in\Theta_{\lambda_{\mathrm{svm}}}\times[-B,B]$,
    \begin{equation}
        \left|
            J(\theta,b;G_{\mathrm{imp}})-J(\theta,b;G)
        \right|
        \le
        \frac{3\xi}{2\lambda_{\mathrm{svm}}}.
        \label{eq:generic-objective-perturbation}
    \end{equation}

    Since $(\widetilde\theta_N,\widetilde b_N)$ minimizes
    $J(\,\cdot\,;G_{\mathrm{imp}})$, applying
    Eq.~\eqref{eq:generic-objective-perturbation} twice gives
    \begin{equation}
    \begin{aligned}
        J(\widetilde\theta_N,\widetilde b_N;G)
        &\le
        J(\widetilde\theta_N,\widetilde b_N;G_{\mathrm{imp}})
        +
        \frac{3\xi}{2\lambda_{\mathrm{svm}}}\\
        &=
        \inf_{\substack{
            \theta\in\Theta_{\lambda_{\mathrm{svm}}}\\
            |b|\le B
        }}
        J(\theta,b;G_{\mathrm{imp}})
        +
        \frac{3\xi}{2\lambda_{\mathrm{svm}}}\\
        &\le
        \inf_{\substack{
            \theta\in\Theta_{\lambda_{\mathrm{svm}}}\\
            |b|\le B
        }}
        J(\theta,b;G)
        +
        \frac{3\xi}{\lambda_{\mathrm{svm}}}.
    \end{aligned}
        \label{eq:generic-exact-objective-bound}
    \end{equation}

    By Eq.~\eqref{eq:generic-unit-margin-separator}, $(f_0,b_0)$ has zero empirical
    hinge loss almost surely, so that
    $\mathcal J(f_0,b_0)=\lambda_{\mathrm{svm}}\|f_0\|_{\mathcal H_k}^2/2\le\lambda_{\mathrm{svm}}R^2/2$
    almost surely.
    For each fixed $b$, apply the same Lagrange-multiplier argument as in Lemma~\ref{lem:svm-dual} to $\mathcal J(\,\cdot\,,b)/\lambda_{\mathrm{svm}}$, with $C=1/(\lambda_{\mathrm{svm}}N)$.
    Since $b$ is fixed, no stationarity condition in $b$ is imposed.
    At a minimizer $f_b$, stationarity in $f$ and the slack variables gives
    \begin{equation*}
        f_b
        =
        \sum_{i=1}^N \gamma_iY_i k(\,\cdot\,,X_i),
        \qquad
        0\le\gamma_i\le C,
    \end{equation*}
    where $\gamma_i$ are the Lagrange multipliers for the margin constraints.
    Thus, the coefficients $\theta_i=\gamma_iY_i$ satisfy $\theta\in\Theta_{\lambda_{\mathrm{svm}}}$.
    Since the minimum of $\mathcal J(\,\cdot\,,b)$ over $\mathcal H_k$
    equals $\min_{\theta\in\Theta_{\lambda_{\mathrm{svm}}}}J(\theta,b;G)$
    % by Lemmas~\ref{lem:representer} and~\ref{lem:svm-dual}
    by this representation
    and Eq.~\eqref{eq:generic-svm-objectives}, we obtain
    \begin{equation*}
        \inf_{\substack{
            \theta\in\Theta_{\lambda_{\mathrm{svm}}}\\
            |b|\le B
        }}
        J(\theta,b;G)
        =
        \inf_{\substack{
            f\in\mathcal H_k\\
            |b|\le B
        }}
        \mathcal J(f,b)
        \le
        \mathcal J(f_0,b_0)
        \le
        \frac{\lambda_{\mathrm{svm}}R^2}{2}
        \qquad
        \text{almost surely}.
    \end{equation*}
    Combining this inequality with
    Eq.~\eqref{eq:generic-exact-objective-bound} gives
    $J(\widetilde\theta_N,\widetilde b_N;G)\le\lambda_{\mathrm{svm}}R^2/2+3\xi/\lambda_{\mathrm{svm}}$.
    By Eqs.~\eqref{eq:generic-svm-objectives}
    and~\eqref{eq:generic-reference-score},
    \begin{equation*}
        J(\widetilde\theta_N,\widetilde b_N;G)
        =
        \frac{\lambda_{\mathrm{svm}}}{2}
        \Bigl\|\sum_{i=1}^N\widetilde\theta_{N,i}k(\,\cdot\,,X_i)\Bigr\|_{\mathcal H_k}^2
        +
        \frac1N\sum_{i=1}^N\max\{1-Y_ig_N(X_i),0\}.
    \end{equation*}
    Since both terms are nonnegative, each is bounded by
    $\lambda_{\mathrm{svm}}R^2/2+3\xi/\lambda_{\mathrm{svm}}$, which gives
    \begin{equation*}
    \begin{aligned}
        \frac1N\sum_{i=1}^N
        \max\{1-Y_i g_N(X_i),0\}
        &\le
        \frac{\lambda_{\mathrm{svm}}R^2}{2}
        +
        \frac{3\xi}{\lambda_{\mathrm{svm}}},\\
        \left\|
            \sum_{i=1}^N
            \widetilde\theta_{N,i}k(\,\cdot\,,X_i)
        \right\|_{\mathcal H_k}^2
        &\le
        R^2+
        \frac{6\xi}{\lambda_{\mathrm{svm}}^2}
        =
        R_\xi^2.
    \end{aligned}
    \end{equation*}

    % The set of scores $x\mapsto f(x)+b$ with
    % $\|f\|_{\mathcal H_k}\le R_\xi$ and $|b|\le B$ has expected
    % Rademacher complexity at most $(R_\xi+B)/\sqrt N$
    % by Lemma~\ref{lem:rkhs-rademacher} and Lemma~\ref{lem:rademacher-properties}(ii).
    % Since $k(x,x)\le1$, each such score satisfies
    % $|f(x)+b|\le R_\xi+B$, so its hinge loss takes values in
    % $[0,1+R_\xi+B]$.
    By Lemma~\ref{lem:rkhs-rademacher} and Lemma~\ref{lem:rademacher-properties}(i)--(ii),
    \begin{equation*}
    \mathfrak R_N\left(
        \left\{
            (x,y)\mapsto\max\{1-y(f(x)+b),0\}:
            \|f\|_{\mathcal H_k}\le R_\xi,\ |b|\le B
        \right\}
    \right)
    \le \frac{R_\xi+B}{\sqrt N}.
\end{equation*}
    Since $k(x,x)\le1$, the reproducing property gives
    \begin{equation*}
        |f(x)+b| \le \|f\|_{\mathcal H_k}\sqrt{k(x,x)}+|b| \le R_\xi+B.
    \end{equation*}
    Hence, for every $y\in\{-1,+1\}$,
    \begin{equation*}
        0 \le \max\{1-y(f(x)+b),0\} \le 1+R_\xi+B.
    \end{equation*}
    
    % These two facts imply, by 
    Combining these observations, we conclude from Lemma~\ref{lem:ulln}
    that, with probability at least $1-\delta$ over
    the random draw of $D_N$, simultaneously for all such $f$ and $b$,
    \begin{equation*}
        R_{\mathrm{hinge}}(x\mapsto f(x)+b)
        \le
        \frac1N\sum_{i=1}^N
        \max\{1-Y_i(f(X_i)+b),0\}
        +
        \Delta_N(R_\xi,B,\delta).
    \end{equation*}
    For every such $G_{\mathrm{imp}}$, the norm bound above together with
    $|\widetilde b_N|\le B$ ensures that $g_N$ is among these scores.
    The uniform bound therefore applies simultaneously to every such $g_N$.
    Combining it with the empirical hinge loss bound gives Eq.~\eqref{eq:generic-population-risk-bound}.
\end{proof}

\subsubsection{Training Gram matrix from finite measurements}
\label{app:training-finite-shot-gram}

Our algorithm projects the training Gram matrix estimated from finite measurements onto PSD matrices, as is done in the experiment (Sec.~\ref{sppl:DiagramOfKernel}), so that the minimization in Eq.~\eqref{eq:generic-svm-minimizer} is a convex quadratic program.
We denote the resulting matrix by $G_{\mathrm{imp}}$ and bound its entrywise error relative to the exact Gram matrix $G$, as required to apply Lemma~\ref{lem:generic-approx-gram-svm}.
The bound follows from Hoeffding's inequality for the fidelity estimates, the Lipschitz continuity of the RBF transform, and the fact that the projection used to construct $G_{\mathrm{imp}}$ does not increase the distance to $G$ in Frobenius norm.

We follow the notation of Sec.~\ref{app:generic-approx-gram-svm}.
In the rest of this section, we consider quantum states as inputs: we take $\mathcal X$ to be a set of density operators of pure states, so that each input $x\in\mathcal X$, including each training input $X_i$, is a pure state.
For $x,x'\in\mathcal X$, let $F(x,x'):=\operatorname{Tr}(xx')$ denote their fidelity.
For a hyperparameter $\sigma>0$, define $h_\sigma$ for fidelity values $q\in[0,1]$ and the kernel $K_\sigma^{\mathrm{RBF}}$ by
\begin{equation}
    h_\sigma(q)
    :=
    \exp\left[
        -\frac{1-q}{\sigma^2}
    \right],
    \qquad
    K_\sigma^{\mathrm{RBF}}(x,x')
    :=
    h_\sigma\bigl(F(x,x')\bigr),
    \label{eq:rbf-kernel-definition}
\end{equation}
which is the quantum RBF kernel in Eq.~\eqref{eq:RBFkernel}.
We use $h_\sigma$ to convert fidelity estimates into RBF kernel estimates.
We take $k=K_\sigma^{\mathrm{RBF}}$ and set $F_{ij}:=F(X_i,X_j)$.
In this setting, the exact Gram matrix $G$ becomes
\begin{equation}
    G_{ij}=h_\sigma(F_{ij}).
    \label{eq:training-exact-fidelity-and-gram}
\end{equation}

We estimate the fidelities $F_{ij}$ from a finite number of measurements.
For each unordered pair $1\le i\le j\le N$, we repeat a binary measurement on copies of $X_i$ and $X_j$ $S_{\mathrm{train}}$ times, and denote by $W_{ij,s}$ the outcome of the $s$-th repetition.
We assume that $W_{ij,1},\ldots,W_{ij,S_{\mathrm{train}}}$ are conditionally independent given $D_N$ and satisfy
\begin{equation}
    W_{ij,s}\in\{0,1\},
    \qquad
    \mathbb E[W_{ij,s}\mid D_N]=F_{ij}.
    \label{eq:training-shot-model}
\end{equation}
Such measurements can be realized by the compute-uncompute method~\cite{havlicek2019supervised}, in which $W_{ij,s}=1$ if the projection onto $X_j$ succeeds and $W_{ij,s}=0$ otherwise.
More generally, if each measurement outcome used for training or prediction takes values in an interval of length $w=O(1)$, as for the SWAP test with outcomes $\pm1$~\cite{PhysRevLett.87.167902}, all the arguments below apply after clipping the sample means $\overline F_{ij}$ in Eq.~\eqref{eq:training-fidelity-estimates} and $\widehat F_i(x)$ to $[0,1]$, with the numbers of measurements multiplied by $w^2$.
For $1\le i\le j\le N$, define the fidelity estimate $\overline F_{ij}$ and the corresponding RBF kernel estimate $\widetilde G_{ij}$ by
\begin{equation}
\begin{aligned}
    \overline F_{ij}
    &:=
    \frac{1}{S_{\mathrm{train}}}
    \sum_{s=1}^{S_{\mathrm{train}}}W_{ij,s},\\
    \widetilde G_{ij}
    &:=
    h_\sigma(\overline F_{ij}),
\end{aligned}
    \label{eq:training-fidelity-estimates}
\end{equation}
and extend $\overline F$ and $\widetilde G$ symmetrically to $i>j$.

The empirical matrix $\widetilde G$ is not necessarily PSD.
Let $\mathcal C_N$ be a closed convex set of symmetric PSD $N\times N$ matrices that contains the exact Gram matrix $G$ for every training sample; examples include the set of all symmetric PSD matrices and its intersection with $\{H:\ H_{ii}=1\ \text{for all }i\}$ or $\{H:\ |H_{ij}|\le1\ \text{for all }i,j\}$.
Define the projected training Gram matrix by
\begin{equation}
    G_{\mathrm{imp}}
    :=
    \operatorname*{arg\,min}_{H\in\mathcal C_N}
    \frac12\|H-\widetilde G\|_F^2.
    \label{eq:training-implemented-gram}
\end{equation}
Since $\mathcal C_N$ is nonempty, closed, and convex, $G_{\mathrm{imp}}$ is uniquely defined.

\begin{lemma}[Gram approximation from finite measurements]
    \label{lem:training-finite-shot-gram}
    Fix a failure probability parameter $\delta_{\mathrm{shot}}\in(0,1)$ and an entrywise error tolerance $\xi>0$.
    If
    \begin{equation}
        S_{\mathrm{train}}
        \ge
        \frac{
            N^2
        }{
            2\sigma^4\xi^2
        }
        \log\left(
            \frac{N(N+1)}{\delta_{\mathrm{shot}}}
        \right),
        \label{eq:training-shots-rbf}
    \end{equation}
    then, for almost every training sample $D_N$,
    \begin{equation}
        \Pr\left(
            \max_{i,j}
            \left|
                (G_{\mathrm{imp}})_{ij}-G_{ij}
            \right|
            \le\xi
            \,\middle|\,
            D_N
        \right)
        \ge
        1-\delta_{\mathrm{shot}}.
        \label{eq:training-gram-probability}
    \end{equation}
\end{lemma}

\begin{proof}
    For pure states,
    \begin{equation*}
        \|x-x'\|_{\mathrm{HS}}^2
        =
        2\bigl[1-F(x,x')\bigr],
    \end{equation*}
    where $\|\cdot\|_{\mathrm{HS}}$ denotes the Hilbert--Schmidt norm.
    Hence $K_\sigma^{\mathrm{RBF}}$ is the Gaussian kernel evaluated on these density operators and is PSD by Schoenberg's theorem~\cite{Schoenberg1938}.
    By the definition of $F$, $K_\sigma^{\mathrm{RBF}}$ is symmetric and satisfies
    \begin{equation*}
        e^{-1/\sigma^2}\le K_\sigma^{\mathrm{RBF}}(x,x')\le1,
        \qquad
        K_\sigma^{\mathrm{RBF}}(x,x)=1.
    \end{equation*}
    By the assumption on $\mathcal C_N$, $G\in\mathcal C_N$, so its metric projection onto $\mathcal C_N$ is $G$ itself.
    Conditional on $D_N$, applying Hoeffding's inequality to both tails
    \cite{Hoeffding1963} and then taking a union bound over
    the $N(N+1)/2$ measured pairs gives, for every $\varepsilon_F>0$,
    \begin{equation}
        \Pr\left(
            \max_{i\le j}
            |\overline F_{ij}-F_{ij}|>\varepsilon_F
            \,\middle|\,
            D_N
        \right)
        \le
        N(N+1)
        \exp\left[
            -2S_{\mathrm{train}}\varepsilon_F^2
        \right].
        \label{eq:training-fidelity-union-bound}
    \end{equation}
    Since $|h_\sigma'(q)|\le\sigma^{-2}$ on $[0,1]$, the mean-value theorem gives
    \begin{equation}
        \left|
            h_\sigma(q)-h_\sigma(q')
        \right|
        \le
        \frac{1}{\sigma^2}
        |q-q'|
        \qquad
        (q,q'\in[0,1]).
        \label{eq:rbf-lipschitz-bound}
    \end{equation}
    Because $F_{ij},\overline F_{ij}\in[0,1]$, the Lipschitz bound in Eq.~\eqref{eq:rbf-lipschitz-bound} and the nonexpansiveness of the metric projection combine to give
    \begin{equation}
    \begin{aligned}
        \max_{i,j}
        |(G_{\mathrm{imp}})_{ij}-G_{ij}|
        &\le
        \|G_{\mathrm{imp}}-G\|_F\\
        &\le
        \|\widetilde G-G\|_F\\
        &\le
        \frac{N}{\sigma^2}
        \max_{i\le j}
        |\overline F_{ij}-F_{ij}|.
    \end{aligned}
        \label{eq:training-projection-bound}
    \end{equation}
    Here, the second inequality follows from the nonexpansiveness of the metric projection onto $\mathcal C_N$ in Frobenius norm, since this projection fixes $G$.
    With $\varepsilon_F=\sigma^2\xi/N$,
    Eq.~\eqref{eq:training-fidelity-union-bound} and the condition in
    Eq.~\eqref{eq:training-shots-rbf} give a failure probability of at most
    $\delta_{\mathrm{shot}}$.
\end{proof}

\subsubsection{Existence of a separator with unit margin}
\label{app:margin-certificate}

Lemma~\ref{lem:generic-approx-gram-svm} assumes the existence of a separator $(f_0,b_0)$ with unit margin, Eq.~\eqref{eq:generic-unit-margin-separator}.
The resulting bound on the expected hinge loss, Eq.~\eqref{eq:generic-population-risk-bound}, depends on this separator only through the bounds $R$ on $\|f_0\|_{\mathcal H_k}$ and $B$ on $|b_0|$.
In this subsection, we show that such a separator exists when, as in Theorem~1, one class has a finite set of representative states such that every state in this class is similar to at least one representative, whereas every state in the other class is dissimilar to all representatives.
The bounds $R$ and $B$ are then inversely proportional to the margin $\eta$ in Theorem~1, which lower-bounds the gap between the two classes in their kernel values averaged over the representatives.

\begin{lemma}[Separator with unit margin]
    \label{lem:anchor-margin-certificate}
    Let $P$ be a probability distribution on
    $\mathcal X\times\{-1,+1\}$ such that $P(Y=y)>0$ for each
    $y\in\{-1,+1\}$.
    Let $k:\mathcal X\times\mathcal X\to[0,1]$ be a measurable symmetric
    PSD kernel with RKHS $\mathcal H_k$.
    For each $y\in\{-1,+1\}$, let
    $\mathcal S_y\subset\mathcal X$ be measurable and satisfy
    \begin{equation*}
        P(X\in\mathcal S_y\mid Y=y)=1.
    \end{equation*}
    Fix $y_\star\in\{-1,+1\}$ and a finite set $A_{y_\star}=\{a_{y_\star,1},\ldots,a_{y_\star,M_{y_\star}}\}\subset\mathcal S_{y_\star}$ with $M_{y_\star}\ge1$.
    Suppose that $\kappa_{y_\star}\in[0,1]$ and
    $\varepsilon_{\mathrm{sep}}\in[0,1)$ satisfy
    \begin{equation}
        \max_{1\le j\le M_{y_\star}}
        k(x,a_{y_\star,j})
        \ge
         \kappa_{y_\star}
        \qquad
         (x\in\mathcal S_{y_\star}),
        \label{eq:margin-anchor-coverage}
    \end{equation}
    and
    \begin{equation}
        k(x,a)
        \le
        \varepsilon_{\mathrm{sep}}
        \qquad
         (x\in\mathcal S_{-y_\star},\ a\in A_{y_\star}).
        \label{eq:margin-anchor-separation}
    \end{equation}
    Define
    \begin{equation}
        \eta
        :=
        \frac{\kappa_{y_\star}}{M_{y_\star}}
        -
        \varepsilon_{\mathrm{sep}}.
        \label{eq:margin-certificate}
    \end{equation}
    If $\eta>0$, then there exist $f_0\in\mathcal H_k$ and
    $b_0\in\mathbb R$ such that
    \begin{equation}
        Y\bigl(f_0(X)+b_0\bigr)
        \ge
        1
        \qquad
        P\text{-almost surely},
        \label{eq:margin-unit-separator}
    \end{equation}
    with
    \begin{equation}
        \|f_0\|_{\mathcal H_k}
        \le
        \frac{2}{\eta},
        \qquad
        |b_0|
        \le
        \frac{2}{\eta}.
        \label{eq:margin-comparator-bounds}
    \end{equation}
\end{lemma}

We refer to Eqs.~\eqref{eq:margin-anchor-coverage} and~\eqref{eq:margin-anchor-separation} as the coverage and separation conditions, respectively.

\begin{proof}
    We prove the claim by explicitly constructing $f_0$ and $b_0$.
    Define
    \begin{equation*}
        \mu
        :=
        \frac{1}{M_{y_\star}}
        \sum_{j=1}^{M_{y_\star}}
        k(\,\cdot\,,a_{y_\star,j}),
        \qquad
        t
        :=
        \frac12
        \left(
            \frac{\kappa_{y_\star}}{M_{y_\star}}
            +
            \varepsilon_{\mathrm{sep}}
        \right).
    \end{equation*}
    The reproducing property, Eqs.~\eqref{eq:margin-anchor-coverage}
    and~\eqref{eq:margin-anchor-separation}, and the condition $k\ge0$ imply
    \begin{equation*}
    \begin{aligned}
        \langle\mu,k(\,\cdot\,,x)\rangle_{\mathcal H_k}
        &\ge
        \frac{\kappa_{y_\star}}{M_{y_\star}}
        &&(x\in\mathcal S_{y_\star}),\\
        \langle\mu,k(\,\cdot\,,x)\rangle_{\mathcal H_k}
        &\le
        \varepsilon_{\mathrm{sep}}
        &&(x\in\mathcal S_{-y_\star}).
    \end{aligned}
    \end{equation*}
    Set
    \begin{equation*}
        f_0
        :=
        \frac{2y_\star}{\eta}\mu,
        \qquad
        b_0
        :=
        -\frac{2y_\star t}{\eta}.
    \end{equation*}
    These choices satisfy
    \begin{equation*}
    \begin{aligned}
        y_\star\bigl(f_0(x)+b_0\bigr)
        &\ge1
        &&(x\in\mathcal S_{y_\star}),\\
        -y_\star\bigl(f_0(x)+b_0\bigr)
        &\ge1
        &&(x\in\mathcal S_{-y_\star}).
    \end{aligned}
    \end{equation*}
    Since $P(X\in\mathcal S_y\mid Y=y)=1$ for each
    $y\in\{-1,+1\}$, these inequalities imply
    Eq.~\eqref{eq:margin-unit-separator}.

    Finally, the triangle inequality and the identity
    $\|k(\,\cdot\,,a)\|_{\mathcal H_k}=\sqrt{k(a,a)}$ give
    \begin{equation*}
        \|\mu\|_{\mathcal H_k}
        \le
        \frac{1}{M_{y_\star}}
        \sum_{j=1}^{M_{y_\star}}
        \sqrt{
            k(a_{y_\star,j},a_{y_\star,j})
        }
        \le1.
    \end{equation*}
    Since also $0\le t\le1$ by its definition,
    Eq.~\eqref{eq:margin-comparator-bounds} follows.
\end{proof}

\subsection{Prediction with finite measurements}
\label{app:prediction-finite-measurements}

Section~\ref{app:training-finite-measurements} bounds the expected hinge loss of the score $g_N$ in Eq.~\eqref{eq:generic-reference-score}.
We bound the misclassification probability when the fidelities used for prediction are estimated from finite measurements.
Fix a positive integer $N$, a kernel hyperparameter $\sigma>0$, and a regularization parameter $\lambda_{\mathrm{svm}}>0$, and consider the quantum RBF kernel $k=K_\sigma^{\mathrm{RBF}}$ in Eq.~\eqref{eq:rbf-kernel-definition}.
% The bound below does not use the optimality of the SVM solution, so we state it for an arbitrary score of the same form as $g_N$.
For this prediction analysis, we consider an arbitrary fixed score of the same form as $g_N$.
Fix pure states $x_1,\ldots,x_N\in\mathcal X$, coefficients $\theta\in\mathbb R^N$ with $\|\theta\|_1\le1/\lambda_{\mathrm{svm}}$, and an intercept $b\in\mathbb R$, and define
\begin{equation}
    g(x)
    :=
    \sum_{i=1}^N
    \theta_i
    h_\sigma\bigl(F(x,x_i)\bigr)
    +
    b.
    \label{eq:prediction-exact-score}
\end{equation}
Let $(X,Y)\sim P$ be a test point.
Let $S_{\mathrm{pred}}$ be a positive integer denoting the number of measurement outcomes collected for each pair consisting of the test state and one of $x_1,\ldots,x_N$.
For each test input $x$ and each $i\in\{1,\ldots,N\}$, let $\widehat F_i(x)$ be the mean of
$S_{\mathrm{pred}}$ independent measurement outcomes in $[0,1]$, each having
expectation $F(x,x_i)$.
Replacing the exact fidelities $F(x,x_i)$ by their estimates gives
\begin{equation}
    \widehat g(x)
    :=
    \sum_{i=1}^N
    \theta_i
    h_\sigma\bigl(\widehat F_i(x)\bigr)
    +
    b.
    \label{eq:prediction-estimated-score}
\end{equation}

\begin{lemma}
    \label{lem:prediction-finite-shot}
    Fix $\varepsilon_{\mathrm{pred}}\in(0,1)$.
    The condition
    \begin{equation}
        S_{\mathrm{pred}}
        \ge
        \frac{2}{\sigma^4\lambda_{\mathrm{svm}}^2}
        \log\left(
            \frac{2N}{\varepsilon_{\mathrm{pred}}}
        \right),
        \label{eq:prediction-shots-rbf}
    \end{equation}
    implies
    \begin{equation}
        \Pr\left(Y\widehat g(X)\le0\right)
        \le
        % 2\varepsilon_{\mathrm{learn}}
        2R_{\mathrm{hinge}}(g)
        +
        \varepsilon_{\mathrm{pred}},
        \label{eq:assembly-prediction-risk}
    \end{equation}
    where the probability is taken over the test point $(X,Y)\sim P$ and the measurements used to obtain $\widehat F_1(X),\ldots,\widehat F_N(X)$.
\end{lemma}

\begin{proof}
    The assumption $\|\theta\|_1\le1/\lambda_{\mathrm{svm}}$ and Eq.~\eqref{eq:rbf-lipschitz-bound} give
    \begin{equation*}
        |\widehat g(x)-g(x)|
        \le
        \frac{1}{\lambda_{\mathrm{svm}}\sigma^2}
        \max_{1\le i\le N}
        |\widehat F_i(x)-F(x,x_i)|.
    \end{equation*}
    Thus the score error is at most $1/2$ whenever the maximum fidelity
    error is at most $\lambda_{\mathrm{svm}}\sigma^2/2$.
    For each fixed $x$, Hoeffding's inequality
    \cite[Theorem~2]{Hoeffding1963} and a union bound over
    $i=1,\ldots,N$ give
    \begin{equation*}
        \Pr\left(
            |\widehat g(x)-g(x)|>\frac12
        \right)
        \le
        2N\exp\left(
            -\frac{
                S_{\mathrm{pred}}
                \lambda_{\mathrm{svm}}^2\sigma^4
            }{2}
        \right),
    \end{equation*}
    where the probability is over the measurements used to predict $x$.
    Because the right-hand side is uniform in $x$, averaging over $X$ and
    applying Eq.~\eqref{eq:prediction-shots-rbf} gives
    \begin{equation*}
        \Pr\left(
            |\widehat g(X)-g(X)|>\frac12
        \right)
        \le
        \varepsilon_{\mathrm{pred}}.
    \end{equation*}
    Since $\mathbf 1\{Yg(X)\le1/2\}\le2\max\{1-Yg(X),0\}$, taking the expectation over $(X,Y)$ gives
    \begin{equation*}
        \Pr\left(Yg(X)\le\frac12\right)
        \le
        2\,\mathbb E_{(X,Y)\sim P}\left[\max\{1-Yg(X),0\}\right]
        =
        2R_{\mathrm{hinge}}(g).
    \end{equation*}
    Misclassification by $\widehat g$ implies that either
    $Yg(X)\le1/2$ or $|\widehat g(X)-g(X)|>1/2$.
    % Taking expectations and using $R_{\mathrm{hinge}}(g_n)\le\varepsilon_{\mathrm{learn}}$ proves the claim.
    Combining these bounds, we obtain
    \begin{equation*}
        \Pr\left(Y\widehat g(X)\le0\right)
        \le
        \Pr\left(Yg(X)\le\frac12\right)
        +
        \Pr\left(|\widehat g(X)-g(X)|>\frac12\right)
        \le
        2R_{\mathrm{hinge}}(g)
        +
        \varepsilon_{\mathrm{pred}},
    \end{equation*}
    which proves Eq.~\eqref{eq:assembly-prediction-risk}.
\end{proof}

\subsection{Formal statement and proof of the main-text theorem}
\label{app:assembly-resource-scaling}

We combine the preceding results to give a formal guarantee for a fixed classification problem in Theorem~\ref{thm:formal-main-polynomial-learning}, from which Theorem~1 follows as shown after its proof.

Let $\eta_0>0$ be a known lower bound on the margin $\eta$ in Eq.~\eqref{eq:margin-certificate}, which is used to choose the hyperparameters below.
For a target misclassification probability $\varepsilon\in(0,1)$ and failure probability $\delta\in(0,1)$, set
\begin{equation*}
    \varepsilon_{\mathrm{learn}}
    :=
    \frac{\varepsilon}{4},
    \qquad
    \varepsilon_{\mathrm{pred}}
    :=
    \frac{\varepsilon}{2},
    \qquad
    \delta_{\mathrm{gen}}
    =
    \delta_{\mathrm{shot}}
    :=
    \frac{\delta}{2}.
\end{equation*}
Here, $\varepsilon_{\mathrm{learn}}$ is the target bound on the expected hinge loss after training, $\varepsilon_{\mathrm{pred}}$ is the additive error term in the misclassification bound of Lemma~\ref{lem:prediction-finite-shot}, and $\delta_{\mathrm{gen}}$ and $\delta_{\mathrm{shot}}$ are the failure probability parameters in Lemmas~\ref{lem:generic-approx-gram-svm} and~\ref{lem:training-finite-shot-gram}, respectively.
Choose $R$, $B$, $\lambda_{\mathrm{svm}}$, and $\xi$ as
\begin{equation}
\begin{aligned}
    R
    &:=
    \frac{2}{\eta_0},
    &
    B
    &:=
    \frac{2}{\eta_0},
    &
    \lambda_{\mathrm{svm}}
    &:=
    \frac{
        \eta_0^2\varepsilon_{\mathrm{learn}}
    }{64},
    &
    \xi
    &:=
    \frac{
        \lambda_{\mathrm{svm}}
        \varepsilon_{\mathrm{learn}}
    }{96}.
\end{aligned}
    \label{eq:assembly-parameter-choice}
\end{equation}

\begin{theorem}[Finite-measurement learning and prediction]
    \label{thm:formal-main-polynomial-learning}
    Let $\mathcal X$ be a measurable space, $P$ a probability distribution on $\mathcal X\times\{-1,+1\}$, and $\sigma>0$.
    Suppose that $K_\sigma^{\mathrm{RBF}}$ in Eq.~\eqref{eq:rbf-kernel-definition} has a separable RKHS and that the hypotheses of
     Lemma~\ref{lem:anchor-margin-certificate} hold with $k=K_\sigma^{\mathrm{RBF}}$.
    Assume further that $\eta$ defined in
    Eq.~\eqref{eq:margin-certificate} satisfies
    \begin{equation}
        \eta
        \ge
        \eta_0.
        \label{eq:formal-margin-lower-bound}
    \end{equation}
    Set
    \begin{equation}
        N
        :=
        \left\lceil
        \frac{4}{\varepsilon_{\mathrm{learn}}^2}
        \left(
            1+
            \frac{2(1+\sqrt2)}{\eta_0}
        \right)^2
        \left(
            2+
            \sqrt{
                \frac{
                    \log(1/\delta_{\mathrm{gen}})
                }{2}
            }
        \right)^2
        \right\rceil,
        \label{eq:assembly-sample-size}
    \end{equation}
    and let $S_{\mathrm{train}}$ and $S_{\mathrm{pred}}$ be the right-hand sides of Eqs.~\eqref{eq:training-shots-rbf} and~\eqref{eq:prediction-shots-rbf}, respectively, rounded up to integers, where the parameters are those in Eq.~\eqref{eq:assembly-parameter-choice} and the value of $N$ above.
    Let $T_{\mathrm{shot}}^{\mathrm{train}}$ and $T_{\mathrm{shot}}^{\mathrm{pred}}$ denote, respectively, the total number of measurement outcomes used to estimate fidelities during training and the number used for one prediction:
    \begin{equation}
    \label{eq:total_measurement_count}
        T_{\mathrm{shot}}^{\mathrm{train}}
        :=
        \frac{N(N+1)}{2}S_{\mathrm{train}},
        \qquad
        T_{\mathrm{shot}}^{\mathrm{pred}}
        :=
        NS_{\mathrm{pred}}.
    \end{equation}
    The training fidelities are estimated as in Sec.~\ref{app:training-finite-shot-gram}: for each pair $1\le i\le j\le N$, $S_{\mathrm{train}}$ binary measurements are performed, with outcomes satisfying Eq.~\eqref{eq:training-shot-model}.
    For prediction, each $\widehat F_i(X)$ is the mean of $S_{\mathrm{pred}}$ measurement outcomes as specified before Lemma~\ref{lem:prediction-finite-shot} with $x_i=X_i$,  obtained 
    % independently of $D_N$ and $(W_{ij,s})$.
    independently of the training measurement outcomes $(W_{ij,s})$, conditional on $D_N$ and $X$.
    Construct $\overline F$, $\widetilde G$, and $G_{\mathrm{imp}}$ as in Eqs.~\eqref{eq:training-fidelity-estimates}--%
    \eqref{eq:training-implemented-gram}.
    For the matrix $G_{\mathrm{imp}}$ defined in
    Eq.~\eqref{eq:training-implemented-gram}, let
    $(\widetilde\theta_N,\widetilde b_N)$ be a measurable minimizer of
    $J(\,\cdot\,;G_{\mathrm{imp}})$ over
    $\Theta_{\lambda_{\mathrm{svm}}}\times[-B,B]$, as in
    Eq.~\eqref{eq:generic-svm-minimizer}.
    Let $\widehat g_N$ be the estimated score in Eq.~\eqref{eq:prediction-estimated-score} with $x_i=X_i$, $\theta=\widetilde\theta_N$, and $b=\widetilde b_N$.
    Then, with probability at least $1-\delta$ over the training sample and the measurements used during training,
    \begin{equation}
        \Pr\left(Y\widehat g_N(X)\le0\right)
        \le
        \varepsilon,
        \label{eq:formal-main-classification-guarantee}
    \end{equation}
    where the probability in Eq.~\eqref{eq:formal-main-classification-guarantee} is conditional on $D_N$ and $(W_{ij,s})$ and is taken over an independent test point $(X,Y)\sim P$ and the measurements used for prediction.
    Moreover, the sample size and measurement counts chosen above to achieve Eq.~\eqref{eq:formal-main-classification-guarantee} satisfy
    \begin{equation}
    \begin{aligned}
        N
        &=
        \widetilde O\left(
            \eta_0^{-2}
            \varepsilon^{-2}
        \right),\\
        T_{\mathrm{shot}}^{\mathrm{train}}
        &=
        \widetilde O\left(
            \sigma^{-4}
            \eta_0^{-12}
            \varepsilon^{-12}
        \right),\\
        T_{\mathrm{shot}}^{\mathrm{pred}}
        &=
        \widetilde O\left(
            \sigma^{-4}
            \eta_0^{-6}
            \varepsilon^{-4}
        \right).
    \end{aligned}
        \label{eq:formal-main-resource-scaling}
    \end{equation}
    Here, $\widetilde O$ suppresses logarithmic factors in $\eta_0^{-1}$, $\varepsilon^{-1}$, and $\delta^{-1}$.
\end{theorem}

\begin{proof}
    We first bound the expected hinge loss of $g_N$, the reference score in Eq.~\eqref{eq:generic-reference-score} with $k=K_\sigma^{\mathrm{RBF}}$ in Eq.~\eqref{eq:rbf-kernel-definition}.
    Lemma~\ref{lem:anchor-margin-certificate} gives a comparator with unit margin satisfying
    $\|f_0\|_{\mathcal H_{K_\sigma^{\mathrm{RBF}}}}\le2/\eta\le R$ and
    $|b_0|\le2/\eta\le B$.
    The parameter choices in Eq.~\eqref{eq:assembly-parameter-choice} give
    \begin{equation}
        \frac{\lambda_{\mathrm{svm}}R^2}{2}
        +
        \frac{3\xi}{\lambda_{\mathrm{svm}}}
        =
        \frac{\varepsilon_{\mathrm{learn}}}{16},
        \qquad
        R_\xi+B
        =
        \frac{2(1+\sqrt2)}{\eta_0}.
        \label{eq:assembly-risk-budgets}
    \end{equation}
    The choice of $N$ in
    Eq.~\eqref{eq:assembly-sample-size} then ensures
    \begin{equation*}
        \Delta_N(R_\xi,B,\delta_{\mathrm{gen}})
        \le
        \frac{\varepsilon_{\mathrm{learn}}}{2}.
    \end{equation*}

    The law of total probability applied to Eq.~\eqref{eq:training-gram-probability} gives $\Pr(\max_{i,j}|(G_{\mathrm{imp}})_{ij}-G_{ij}|\le\xi)\ge1-\delta_{\mathrm{shot}}$, where the probability is joint over the training sample and the measurements used for training. 
    Lemma~\ref{lem:generic-approx-gram-svm} gives an event of probability at least $1-\delta_{\mathrm{gen}}$ over the training sample on which the conclusion of Lemma~\ref{lem:generic-approx-gram-svm} holds uniformly over all measurement realizations satisfying $\max_{i,j}|(G_{\mathrm{imp}})_{ij}-G_{ij}|\le\xi$.
    % the measurement realizations.
    On the intersection of these events,
    \begin{equation}
        R_{\mathrm{hinge}}(g_N)
        \le
        \left(
            \frac1{16}
            +
            \frac12
        \right)
        \varepsilon_{\mathrm{learn}}
        <
        \varepsilon_{\mathrm{learn}}.
    \label{eq:assembly-training-risk}
    \end{equation}
    A union bound shows that this intersection has probability at least $1-\delta_{\mathrm{gen}}-\delta_{\mathrm{shot}}=1-\delta$.
    For every training realization in this intersection, $\|\widetilde\theta_N\|_1\le1/\lambda_{\mathrm{svm}}$ by Eq.~\eqref{eq:generic-theta-box}, and Lemma~\ref{lem:prediction-finite-shot} with $x_i=X_i$, $\theta=\widetilde\theta_N$, and $b=\widetilde b_N$, for which $g=g_N$, gives Eq.~\eqref{eq:formal-main-classification-guarantee}, since $2R_{\mathrm{hinge}}(g_N)+\varepsilon_{\mathrm{pred}}\le2\varepsilon_{\mathrm{learn}}+\varepsilon_{\mathrm{pred}}=\varepsilon$.

    We next derive the resource bounds.
    We first show that $\sigma^{-4}\eta_0^{-4}\ge1$.
    The bound $K_\sigma^{\mathrm{RBF}}(x,a)\ge e^{-1/\sigma^2}$ and the separation condition in Eq.~\eqref{eq:margin-anchor-separation} imply $\varepsilon_{\mathrm{sep}}\ge e^{-1/\sigma^2}$.
    The definition in Eq.~\eqref{eq:margin-certificate}, the assumption $\eta_0\le\eta$, and $\kappa_{y_\star}/M_{y_\star}\le1$ therefore give $\eta_0\le\eta\le1-e^{-1/\sigma^2}\le\min\{1,\sigma^{-2}\}$.
    Hence $\sigma^{-4}\eta_0^{-4}\ge\eta_0^{-2}\ge1$, so rounding the sample size and measurement counts up to integers does not change the asymptotic bounds below.
    The definitions of the total measurement counts in Eq.~\eqref{eq:total_measurement_count} and the choices of $S_{\mathrm{train}}$ and $S_{\mathrm{pred}}$ from Eqs.~\eqref{eq:training-shots-rbf} and~\eqref{eq:prediction-shots-rbf} give
    \begin{equation*}
    \begin{aligned}
        T_{\mathrm{shot}}^{\mathrm{train}}
        &=
        \widetilde O\left(
            \frac{N^4}{\sigma^4\xi^2}
        \right),\\
        T_{\mathrm{shot}}^{\mathrm{pred}}
        &=
        \widetilde O\left(
            \frac{N}{\sigma^4\lambda_{\mathrm{svm}}^2}
        \right).
    \end{aligned}
    \end{equation*}
    % For these choices of $N$, $S_{\mathrm{train}}$, and $S_{\mathrm{pred}}$, 
    Substituting the choices of $N$, $\lambda_{\mathrm{svm}}$, and $\xi$ above into these bounds, we obtain, up to logarithmic factors,
    \begin{equation}
    \begin{aligned}
        N
        &=
        \widetilde O\left(
            \eta_0^{-2}
            \varepsilon_{\mathrm{learn}}^{-2}
        \right),\\
        T_{\mathrm{shot}}^{\mathrm{train}}
        &=
        \widetilde O\left(
            \sigma^{-4}
            \eta_0^{-12}
            \varepsilon_{\mathrm{learn}}^{-12}
        \right),\\
        T_{\mathrm{shot}}^{\mathrm{pred}}
        &=
        \widetilde O\left(
            \sigma^{-4}
            \eta_0^{-6}
            \varepsilon_{\mathrm{learn}}^{-4}
        \right).
    \end{aligned}
        \label{eq:assembly-resource-scaling}
    \end{equation}
    Substituting $\varepsilon_{\mathrm{learn}}=\varepsilon/4$ in Eq.~\eqref{eq:assembly-resource-scaling} gives Eq.~\eqref{eq:formal-main-resource-scaling}.
\end{proof}

% [御手洗：theorem 自体には system size を入れない形に書き換え、ここでそれらの仮定を全て吸収する形に整理。theorem は \(\eta\) の lower bound についてコストが多項式になる、ということが本質的なところで、\(\eta\) がシステムサイズの多項式になる、とかそういうことは本質的でないため。]

Theorem~1 in the main text follows from Theorem~\ref{thm:formal-main-polynomial-learning} applied to a family of classification problems indexed by an integer $m\ge1$ representing the system size.
Since the margin $\eta$ depends on $\sigma$ through the kernel, the choice of $\sigma$ is constrained by requiring the margin condition
\begin{equation}
    \eta
    \ge
    \eta_0
    =
    \Omega(m^{-\nu})
    \label{eq:assembly-margin-lower-bound}
\end{equation}
for a constant $\nu\ge0$ independent of $m$: if $\sigma$ is too small, the kernel values between states and their representatives become too small for the coverage condition with $\kappa_{y_\star}\ge\eta_0$.
As a result, any $\sigma$ for which Eq.~\eqref{eq:assembly-margin-lower-bound} holds satisfies $\sigma=\Omega(1/\sqrt{1+\log m})$, provided that the class $y_\star$ does not collapse onto its representatives in the sense that some $x\in\mathcal S_{y_\star}$ satisfies
\begin{equation}
    F(x,a)
    \le
    1-c
    \qquad
    (a\in A_{y_\star})
    \label{eq:no-collapse}
\end{equation}
with a constant $c>0$ independent of $m$.
Indeed, for such $x$, the coverage condition in Eq.~\eqref{eq:margin-anchor-coverage} and $\kappa_{y_\star}\ge\eta\ge\eta_0$, which follows from the definition of $\eta$ in Eq.~\eqref{eq:margin-certificate}, give $\eta_0\le\max_{a\in A_{y_\star}}K_\sigma^{\mathrm{RBF}}(x,a)\le\exp(-c/\sigma^2)$, so that $\sigma^2\ge c/\log(1/\eta_0)=\Omega(1/(1+\log m))$.
Therefore, under Eqs.~\eqref{eq:assembly-margin-lower-bound} and~\eqref{eq:no-collapse}, Eq.~\eqref{eq:formal-main-resource-scaling} gives
\begin{equation*}
    N
    =
    \widetilde O\left(m^{2\nu}\varepsilon^{-2}\right),
    \qquad
    T_{\mathrm{shot}}^{\mathrm{train}}
    =
    \widetilde O\left(m^{12\nu}\varepsilon^{-12}\right),
    \qquad
    T_{\mathrm{shot}}^{\mathrm{pred}}
    =
    \widetilde O\left(m^{6\nu}\varepsilon^{-4}\right),
\end{equation*}
which are polynomial in $m$ for fixed $\varepsilon$ and $\delta$.

\section{VERIFICATION OF COVERAGE, SEPARATION, AND MARGIN FOR STELLAR RANK CLASSIFICATION}\label{sppl:CheckOfAssumptions}

For the ideal state families defining the stellar-rank classification tasks considered in the main text, we verify the coverage and separation conditions for the exact quantum RBF kernel in Eqs.~\eqref{eq:margin-anchor-coverage} and \eqref{eq:margin-anchor-separation}.
We then show that the resulting margin satisfies the inverse-polynomial lower bound in Eq.~\eqref{eq:assembly-margin-lower-bound}.

\begin{figure}[htbp]
    \centering
    \includegraphics[width=0.9\linewidth]{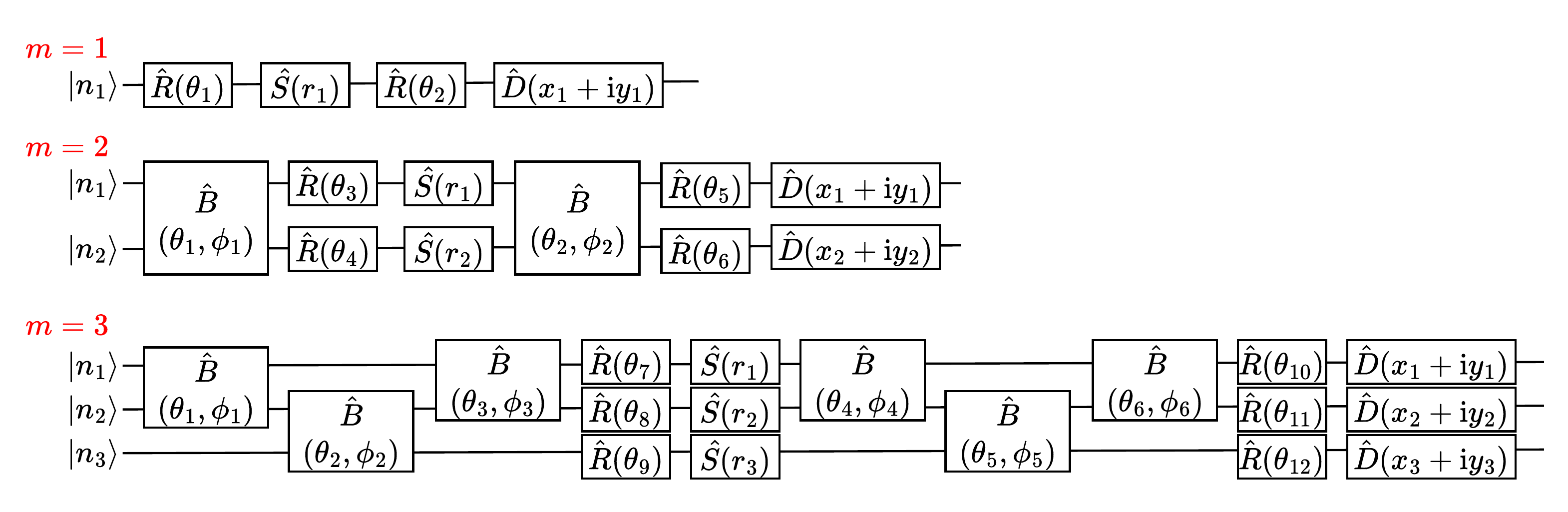}
    \caption{\textbf{Quantum circuit used to generate multimode quantum data.} The simulations use states with one, two, and three modes.}
    \label{fig:MultiModeCircuit}
\end{figure}

\subsection{Classification tasks and state families}

\subsubsection{Multimode state family}

The datasets used in the simulations reported in Fig.~\ref{fig:MultiModeAccuracy_main} are constructed from the following state family. For any integer $m\geq1$, the family consists of pure $m$-mode states of the form
\[
\ket{\psi}=\hat{G}\ket{\bm{n}},
\]
where $\hat{G}$ is an $m$-mode Gaussian unitary and $\ket{\bm{n}}$ is the product Fock state
\[
\ket{\bm{n}}=\ket{n_1}\otimes\cdots\otimes\ket{n_m},
\qquad
\bm{n}=(n_1,\ldots,n_m)\in\mathbb{Z}_{\geq 0}^m.
\]
A general Gaussian unitary operator $\hat{G}$ can be decomposed into a beam-splitter network, single-mode squeezing operations, another beam-splitter network, and displacement operations by the Bloch--Messiah decomposition~\cite{PhysRevA.71.055801}. 
Accordingly, these datasets are generated using the circuit shown in Fig.~\ref{fig:MultiModeCircuit}. 
The beam splitter and phase rotation are defined by
\begin{align*}
    \hat{B}(\theta,\phi)
    &=
    \exp\!\left[
        \theta\left(
            e^{\mathrm{i}\phi}\hat{a}\otimes\hat{a}^\dagger
            -e^{-\mathrm{i}\phi}\hat{a}^\dagger\otimes\hat{a}
        \right)
    \right],\\
    \hat{R}(\theta)
    &=
    \exp\!\left(\mathrm{i}\theta\hat{a}^\dagger\hat{a}\right).
\end{align*}
The squeezing and displacement operators are defined by
\begin{align*}
    \hat{S}(r)
    &=
    \exp\!\left[
        \frac{r\hat{a}^2-r(\hat{a}^\dagger)^2}{2}
    \right],\\
    \hat{D}(x+\mathrm{i}y)
    &=
    \exp\!\left[
        (x+\mathrm{i}y)\hat{a}^\dagger
        -
        (x-\mathrm{i}y)\hat{a}
    \right].
\end{align*}
Here, $\theta,\phi,r,x,y\in\mathbb{R}$. 
The operators $\hat{a}$ and $\hat{a}^\dagger$ denote the bosonic annihilation and creation operators, respectively.
We collect the circuit parameters in
\[
\bm{\vartheta}=(\bm{n},\bm{\theta},\bm{\phi},\bm{r},\bm{x},\bm{y}),
\]
with
\begin{align*}
\bm{n}&=(n_1,\ldots,n_m)\in\mathbb{Z}_{\ge 0}^m, \\
\bm{\theta}&=(\theta_1,\ldots,\theta_{m(m+1)})\in[0,2\pi)^{m(m+1)}, \\
\bm{\phi}&=(\phi_1,\ldots,\phi_{m(m-1)})\in[0,2\pi)^{m(m-1)}, \\
\bm{r}&=(r_1,\ldots,r_m)\in\mathbb{R}^m, \\
\bm{x}&=(x_1,\ldots,x_m)\in\mathbb{R}^m, \\
\bm{y}&=(y_1,\ldots,y_m)\in\mathbb{R}^m .
\end{align*}
The state specified by $\bm{\vartheta}$ is denoted by
\[
\ket{\psi(\bm{\vartheta})}.
\]
Gaussian unitaries preserve stellar rank, and the product Fock state $\ket{\bm{n}}$ has stellar rank $|\bm{n}|$ \cite{PhysRevResearch.3.033018}. We therefore assign each state the label
\[
z=|\bm{n}|:=\sum_{i=1}^{m} n_i .
\]
In the numerical simulations reported in Fig.~\ref{fig:MultiModeAccuracy_main}, we consider the minimal nontrivial binary task
\[
z\in\{0,1\}.
\]
The construction extends to multiclass classification involving higher stellar ranks; see Supplementary Sec.~\ref{sppl:numericalsimulation}.

Motivated by the sufficient conditions in Lemma~\ref{lem:anchor-margin-certificate}, we construct the dataset in two stages: first a finite set of representative states for each class and then data states whose fidelity with their representatives exceeds a prescribed threshold.

\paragraph*{Representative states.}
For each label $z\in\{0,1\}$, we generate a finite set of representative states
\[
A_z=\{a_{z,1},\ldots,a_{z,M_z}\},
\qquad
a_{z,i}:=\ket{\psi(\bm{\vartheta}^{(i)}_z)},
\]
with
\[
M_z=9m(m+1).
\]
We choose the number of representative states to scale with the number of independently varied parameters in the circuit shown in Fig.~\ref{fig:MultiModeCircuit}. 
The prefactor is chosen so that the single-mode construction uses the same number of representative states per class as the classification task implemented in the experiment. 
Thus, $M_z=\Theta(m^2)$.

Each representative state is defined by the parameter vector
\[
\bm{\vartheta}^{(i)}_z=
\bigl(\bm{n}_z^{(i)},\bm{\theta}_z^{(i)},\bm{\phi}_z^{(i)},\bm{r}_z^{(i)},\bm{x}_z^{(i)},\bm{y}_z^{(i)}\bigr)
\]
whose components are chosen as follows. 
First, the components of $\bm{\theta}_z^{(i)}$ and $\bm{\phi}_z^{(i)}$ are sampled independently and uniformly from $[0,2\pi)$. 
Second, all squeezing levels are fixed at 4 dB. 
Third, the displacement parameters $\bm{x}_z^{(i)}$ and $\bm{y}_z^{(i)}$ have components sampled independently and uniformly from $[-1,1]$. 
Finally, the Fock-state component is chosen according to the class label:
\begin{itemize}
    \item for $z=0$, we set
    \[
    \bm{n}_z^{(i)}=\bm{0};
    \]
    \item for $z=1$, we choose $\bm{n}_z^{(i)}$ uniformly from
    \[
    \{\bm{e}_1,\ldots,\bm{e}_m\}.
    \]
    Here, $\bm{e}_\ell$ is the $\ell$th standard basis vector in
    $\mathbb{Z}^m$.
\end{itemize}

\paragraph*{Data states.}
To generate a candidate state in class $z$, we first sample a representative state $a_{z,i}\in A_z$ uniformly. 
All parameters are held fixed except for the squeezing parameters, which are resampled independently and uniformly over the interval corresponding to 2--6 dB. 
The resulting candidate $\ket{\psi}$ is retained only if its fidelity with the corresponding representative satisfies
\begin{equation}
F\!\left(\ket{\psi},a_{z,i}\right)\ge F_{\mathrm{in}},
\qquad
F_{\mathrm{in}}=0.9.
\label{eq:stellar_Fin_condition}
\end{equation}
Let $\mathcal S_z$ denote the family of all candidate states satisfying Eq.~\eqref{eq:stellar_Fin_condition}.
Because 4 dB lies in the allowed squeezing range and $F(a,a)=1\ge F_{\mathrm{in}}$, every $a\in A_z$ belongs to $\mathcal S_z$.

\subsubsection{Single-mode state family for the experimental task}
\label{subsubsec:single_mode_exp}

We next define the ideal state family underlying the single-mode classification task implemented in the experiment reported in Fig.~\ref{fig:accuracy_exp}.
This family is a discrete single-mode version of the construction above.
For each $z\in\{0,1\}$, let $A_z$ denote the set of representative states specified by
\[
n_1=z,
\qquad
\theta_2\in\left\{0,\frac{\pi}{2}\right\},
\qquad
x_1,y_1\in\{-1,0,1\},
\qquad
r_1=\frac{4\ln(10)}{20}.
\]
The two values of $\theta_2$ and the nine displacement pairs yield $M_z=18$ representatives per class. For this task, let $\mathcal S_z$ denote the ideal class-$z$ state family containing, for each representative, the five states defined by squeezing levels of 2, 3, 4, 5, and 6 dB, with all other parameters fixed. The parameter $\phi$ is absent in this single-mode case, and $\theta_1$ is omitted because $\hat{R}(\theta_1)\ket{n_1}$ differs from $\ket{n_1}$ only by a global phase. 

\subsection{Coverage and separation}\label{subsec:verification_of_assumptions}

We verify the coverage and separation conditions required by Theorem~\ref{thm:formal-main-polynomial-learning}, separately for the two ideal state-family constructions defined above.

Throughout this subsection, $k_\sigma$ denotes $K_\sigma^{\mathrm{RBF}}$ defined in Eq.~\eqref{eq:RBF_sppl}.
We use $Y=2z-1$ as the SVM label for $z\in\{0,1\}$.

Conditioning on the realized representative sets, let $P$ be any distribution on
$\mathcal X\times\{-1,+1\}$.
We assume that the probability of sampling quantum data from each class is nonzero, i.e., for each
$z\in\{0,1\}$,
\begin{equation*}
    P(Y=2z-1)>0.
\end{equation*}
% Moreover, since the class label is defined by the stellar rank, which is uniquely determined for each quantum state, any sampled state with label $Y=2z-1$ belongs to the corresponding state family $\mathcal S_z$. 
Moreover, we assume that, under $P$, data states with label $Y=2z-1$ are sampled from the corresponding state family $\mathcal S_z$ constructed above.
Therefore,
\begin{equation*}
    P(X\in\mathcal S_z\mid Y=2z-1)=1.
\end{equation*}
Together, these conditions establish the distributional assumptions required by Theorem~1; the remaining coverage and separation conditions are verified below.

\subsubsection{Coverage within each class}

For the multimode construction, let $\ket{\psi}\in\mathcal S_z$.
By construction, there is a representative state $a\in A_z$ such that
\begin{equation*}
    F(\ket{\psi},a)
    \ge
    F_{\mathrm{in}}.
\end{equation*}
By definition of the quantum RBF kernel,
\begin{equation}
    \max_{a\in A_z} k_\sigma(\ket{\psi},a)
    \ge
    \exp\!\left[-\frac{1-F_{\mathrm{in}}}{\sigma^2}\right]
    =:
    \kappa_z .
    \label{eq:stellar_kappa}
\end{equation}
Thus, the multimode state families $\mathcal S_z$ and representative sets
$A_z$ satisfy the coverage condition in Eq.~\eqref{eq:margin-anchor-coverage}.

For the state family defining the single-mode classification task used in the experiment, the acceptance threshold $F_{\mathrm{in}}=0.9$ is replaced by $F_{\mathrm{in}}^{\mathrm{exp}}=0.49$ to retain all experimental data points, including those with reduced fidelities due to experimental imperfections.
Accordingly, it suffices to replace $F_{\mathrm{in}}$ by $F_{\mathrm{in}}^{\mathrm{exp}}$ in the previous argument.

\subsubsection{Separation between classes}\label{subsec:separation_between_classes}

To verify the separation condition, we bound the fidelity between states of stellar rank zero and one.
For the state families considered here, invariance of fidelity under Gaussian unitaries and a relabelling of the modes reduce the optimization defining multimode stellar robustness
\cite[Appendix~E, Theorem~2]{PRXQuantum.2.020333} to the comparison between $\ket1\otimes\ket0^{\otimes(m-1)}$ and an $m$-mode Gaussian state.
Tracing out all but the occupied mode yields a single-mode Gaussian state, and monotonicity of fidelity under partial trace~\cite{Liang_2019} bounds the multimode fidelity by the corresponding single-mode optimum. 
Appending vacuum modes gives the converse bound. The multimode optimum therefore equals the single-mode value reported in Ref.~\cite[Sec.~III.B]{PRXQuantum.2.020333}:
\begin{equation*}
    F_{\mathrm{out}}
    :=
    \frac{3\sqrt{3}}{4e}
    \simeq
    0.478.
\end{equation*}
Hence, for any $\ket{\psi}\in\mathcal S_z$ and $a\in A_{1-z}$,
$F(\ket{\psi},a)\le F_{\mathrm{out}}$. 
Consequently,
\begin{equation}
    k_\sigma(\ket{\psi},a)
    \le
    \exp\!\left[-\frac{1-F_{\mathrm{out}}}{\sigma^2}\right]
    =:
    \varepsilon_{\mathrm{sep}}.
    \label{eq:stellar_epssep}
\end{equation}
Thus, the separation condition in Eq.~\eqref{eq:margin-anchor-separation} is satisfied.

\subsection{Inverse-polynomial lower bound on the margin}
For the multimode construction, both classes use the same number of
representative states and the same fidelity threshold $F_{\mathrm{in}}$.
Hence,
\[
M_0=M_1=:M=9m(m+1),
\qquad
\kappa_0=\kappa_1=:\kappa=\exp\!\left[-\frac{1-F_{\mathrm{in}}}{\sigma^2}\right].
\]
In particular, $M=\Theta(m^2)$.
Equations~\eqref{eq:margin-certificate},
\eqref{eq:stellar_kappa}, and
\eqref{eq:stellar_epssep} therefore give
\begin{equation}
\eta
=
\frac{\kappa}{M}-\varepsilon_{\mathrm{sep}} .
\label{eq:stellar_eta}
\end{equation}
Fix a constant $C'>0$ independent of $m$ and set
\begin{equation}
\sigma^2
=
\frac{F_{\mathrm{in}}-F_{\mathrm{out}}}{\ln M + C'}.
\label{eq:stellar_sigma_choice}
\end{equation}
With this choice,
\begin{align*}
\frac{\kappa}{M}
&=
\frac{1}{M}\exp\left[-\frac{1-F_{\mathrm{in}}}{\sigma^2}\right]
\\
&=
M^{-\alpha}
\exp\!\left[-C'\frac{1-F_{\mathrm{in}}}{F_{\mathrm{in}}-F_{\mathrm{out}}}\right],
\\
\varepsilon_{\mathrm{sep}}
&=
\exp\left[-\frac{1-F_{\mathrm{out}}}{\sigma^2}\right]
\\
&=
M^{-\alpha}
\exp\!\left[-C'\frac{1-F_{\mathrm{out}}}{F_{\mathrm{in}}-F_{\mathrm{out}}}\right],
\end{align*}
where
\[
\alpha:=\frac{1-F_{\mathrm{out}}}{F_{\mathrm{in}}-F_{\mathrm{out}}}>0.
\]
Substitution into Eq.~\eqref{eq:stellar_eta} yields
\begin{equation}
\eta
=
M^{-\alpha}
\exp\!\left[-C'\frac{1-F_{\mathrm{out}}}{F_{\mathrm{in}}-F_{\mathrm{out}}}\right]
\left(e^{C'}-1\right).
\label{eq:stellar_eta_explicit}
\end{equation}
Because $F_{\mathrm{in}}>F_{\mathrm{out}}$ and
$F_{\mathrm{in}}$, $F_{\mathrm{out}}$, and $C'$ are independent of $m$,
the factor multiplying $M^{-\alpha}$ is a positive constant. Since
$M=9m(m+1)\le18m^2$ for $m\ge1$,
Eq.~\eqref{eq:stellar_eta_explicit} gives
\begin{equation*}
\eta
\ge
c_\eta m^{-\nu},
\qquad
\nu:=2\alpha,
\qquad
c_\eta
:=
18^{-\alpha}
\exp\!\left[-C'\frac{1-F_{\mathrm{out}}}{F_{\mathrm{in}}-F_{\mathrm{out}}}\right]
\left(e^{C'}-1\right)
>0.
\end{equation*}
Thus Eq.~\eqref{eq:assembly-margin-lower-bound} holds with
$\eta_0:=c_\eta m^{-\nu}$ and $\nu=2\alpha$. Moreover,
Eq.~\eqref{eq:stellar_sigma_choice} gives
$\sigma=\Omega(1/\sqrt{\log m})$. 
Under the bounded measurement model of Theorem~1, the resulting sufficient training and prediction measurement
counts are therefore polynomial in $m$ for fixed risk and confidence parameters.

For the single-mode task associated with the experiment, the same calculation
applies with $F_{\mathrm{in}}^{\mathrm{exp}}=0.49$ in place of
$F_{\mathrm{in}}$. Since
$F_{\mathrm{in}}^{\mathrm{exp}}>F_{\mathrm{out}}$, the same construction
gives a strictly positive margin.

The hyperparameter $\sigma$ constructed above are sufficient theoretical choices. By
contrast, $\sigma$ used in the numerical simulations and the experiment
were selected by cross-validation; the argument above does not establish the
same margin bound for these empirically selected values.

For the theoretical hyperparameter constructed above, this verifies the coverage,
separation, and margin conditions required by Theorem~1 for the ideal state
families underlying the numerical simulations and the experimental
classification task.

\subsection{Experimental fidelity estimation}

The preceding analysis establishes the sufficient conditions for the ideal classification task. The present experimental implementation uses a homodyne-based fidelity estimator, whose statistical model differs from the binary-outcome model considered in Theorem~1. More specifically, the fidelity estimator used in the experiment is constructed from
phase-averaged homodyne data. Because the value assigned to an individual
homodyne outcome can lie outside $[0,1]$, this estimator does not satisfy the
measurement model in Eq.~\eqref{eq:training-shot-model}. This restriction is
specific to the homodyne estimator and is not imposed by the classification
task or the quantum RBF kernel. The measurement model can instead be
satisfied using a bounded unbiased fidelity estimator.

For the state families considered here, such an estimator can be obtained by
applying the exact inverse Gaussian operation associated with the target
state and assigning one to the target Fock outcome and zero to every other
outcome. Under independent repetitions, these binary outcomes satisfy
Eq.~\eqref{eq:training-shot-model}.

A SWAP test provides another bounded estimator, but assigning $+1$ and $-1$
to its two ancilla outcomes gives a random variable in $\{-1,+1\}$ whose
expectation equals the fidelity. 
It is covered by the bounded-outcome extension described after Eq.~\eqref{eq:training-shot-model}.
The empirical mean is then clipped to $[0,1]$.

\section{STELLAR ROBUSTNESS AND ITS APPLICATION TO OUR TASK}\label{sppl:StellarRobustnessAnd}

Here, we describe the non-QML method for identifying the stellar rank and its application in our experiment. Conventionally, several theoretical works provide the witness of stellar rank~\cite{PhysRevLett.124.063605, PRXQuantum.2.020333, PhysRevLett.123.043601, fiurasek2022efficient, hct3-6hgn}. Among them, we adopt the concept of ``stellar robustness~\cite{PhysRevLett.124.063605}'' due to its convenience in our experiment. The stellar robustness $\mathrm{Rob}^\ast(\ket{\psi})$ of a state $\ket{\psi}$ is defined as
\begin{equation}
    \mathrm{Rob}^\ast(\ket{\psi}) = \inf_{r^\ast(\ket{\phi}) < r^\ast_0} D_1(\ket{\phi}, \ket{\psi}),
\end{equation}
where $r^\ast_0=r^\ast(\ket{\psi})$ is the stellar rank of the state $\ket{\psi}$, and $D_1(\ket{\phi}, \ket{\psi})=\sqrt{1-|\braket{\phi|\psi}|^2}$ denotes the trace distance between $\ket{\phi}$ and $\ket{\psi}$. The stellar robustness is also related to the fidelity $F(\ket{\phi}, \ket{\psi}) = |\braket{\phi|\psi}|^2$ by the following relation:
\begin{equation}
    \sup_{r^\ast(\ket{\phi}) < r^\ast_0} F(\ket{\phi}, \ket{\psi}) = 1 - [\mathrm{Rob}^\ast(\ket{\psi})]^2.
    \label{eq:StellarRobustnessFidelity}
\end{equation}
From Eq.~\eqref{eq:StellarRobustnessFidelity}, we see that if a state $\ket{\phi}$ has fidelity greater than $1 - [\mathrm{Rob}^\ast(\ket{\psi})]^2$ with $\ket{\psi}$, then $\ket{\phi}$ must have stellar rank greater than or equal to $r^\ast(\ket{\psi})$.

We estimate the stellar rank of test quantum data and evaluate the accuracy labeled “non-QML” in Figs.~\ref{fig:MultiModeAccuracy_main} and~\ref{fig:accuracy_exp}, following the strategy proposed in Ref.~\cite{PhysRevLett.124.063605}. Reference~\cite{PhysRevLett.124.063605} introduces a method to certify the stellar rank of an unknown quantum state via the stellar robustness. The key idea is to directly measure the fidelity between a test state of unknown stellar rank and reference states of known stellar rank greater than or equal to one, without performing full quantum-state tomography. If there exists at least one reference state whose fidelity with the test state exceeds the threshold, the stellar rank of the test state is certified to be equal to or higher than that of the reference state.

Hereafter, we describe the classification of stellar rank zero and one considered in the main text. The stellar robustness is analytically calculated for some single-mode cases~\cite{PRXQuantum.2.020333}. By contrast, it is difficult to analytically derive it for multi-mode quantum states, and it is generally derived numerically. However, for the classification of multi-mode quantum states with stellar rank zero and one, as we discuss in Supplementary Sec.~\ref{subsec:separation_between_classes}, the relevant fidelity threshold for multi-mode quantum data is identical to that of the corresponding single-mode case. Hence, we determine whether a quantum state has stellar rank one by comparing the fidelity between the test state and a reference state known to have stellar rank one, with the single-mode threshold derived in Ref.~\cite{PRXQuantum.2.020333}. We note that the threshold can be used even if the tested quantum state is a mixed state.

In general, the test state is unknown, and the choice of reference states is nontrivial. One simple strategy is to select the reference states randomly, while a more systematic strategy is to iteratively optimize the reference states to increase their fidelity with the test state. In this work, we implement both strategies. The random strategy is denoted as “non-QML (random)” in Fig.~\ref{fig:MultiModeAccuracy_main} and Fig.~\ref{fig:accuracy_exp}, while the optimization-based strategy is denoted as “non-QML (optimized)” in Fig.~\ref{fig:MultiModeAccuracy_main}. For both strategies, we iteratively prepare reference states of stellar rank one, measure their fidelities with the test state, and assign stellar rank one to the test state if there exists at least one reference state whose fidelity exceeds the threshold value derived in Ref.~\cite{PRXQuantum.2.020333}. For the optimization-based strategy, we benchmark several classical optimization algorithms, such as Bayesian Optimization, Constrained Optimization BY Linear Approximation (COBYLA), Covariance Matrix Adaptation Evolution Strategy, Powell, and Nelder-Mead. Finally, we adopt COBYLA, which achieves the highest performance among the tested methods.

We note that the optimization-based strategy is not experimentally implemented in Fig.~\ref{fig:accuracy_exp}, as it requires an iterative coordination between quantum and classical processors, which is beyond the scope of the present experiment. In contrast, the random strategy can be straightforwardly implemented without additional experiments once the pairwise fidelities $|\braket{\psi_j|\psi_i}|^2$ have been measured for the quantum kernel method (Fig.~\ref{fig:kernel_exp}). Specifically, the random strategy checks whether the fidelity between a test state and at least one randomly chosen reference state exceeds a fixed threshold; these fidelities $|\braket{\psi_j|\psi_i}|^2$ are already measured to construct the kernel. Therefore, we can evaluate the random strategy directly from the same measured fidelity data used for QML, by treating a subset of the training states as candidate reference states and comparing their fidelities with the test state to the threshold.

\section{NUMERICAL SIMULATION OF THE MULTICLASS CLASSIFICATION}\label{sppl:numericalsimulation} 

\begin{figure}[htbp]
    \centering
    \includegraphics[width=0.4\linewidth]{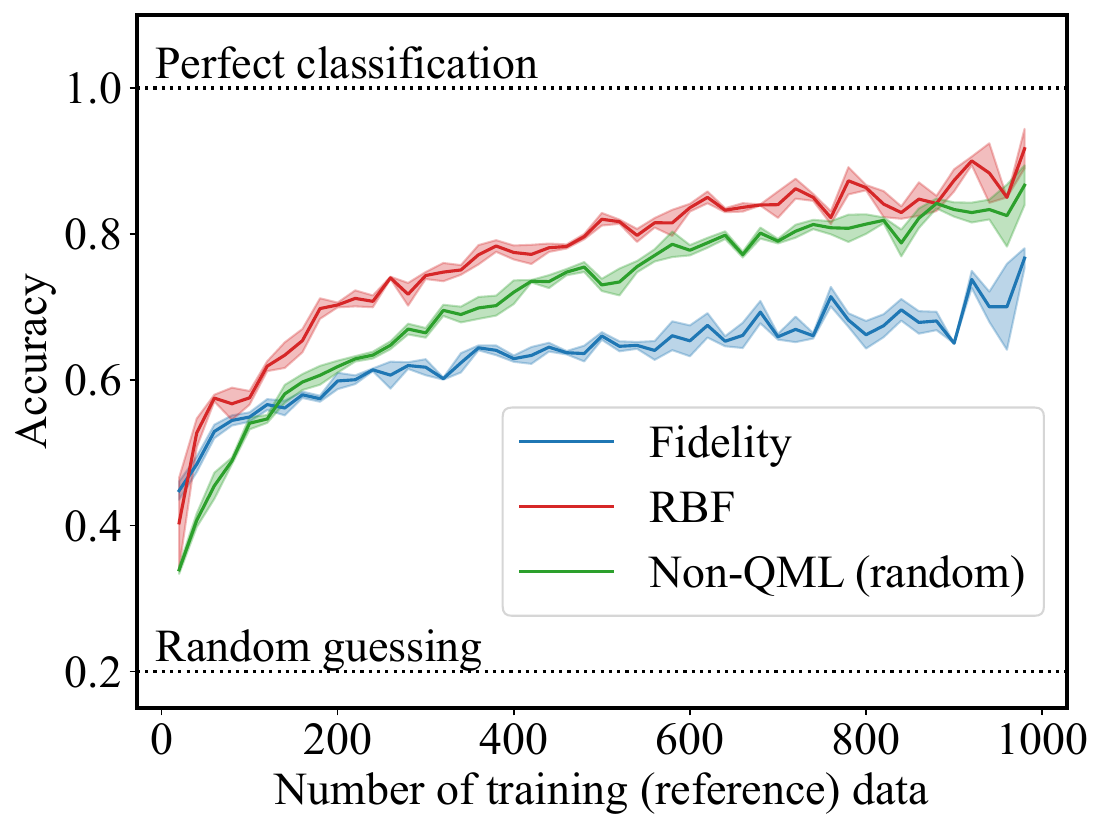}
    \caption{\textbf{Numerical results of classification accuracy for higher stellar rank.} Blue, red, and green lines correspond to results obtained using the quantum fidelity kernel, quantum RBF kernel, and non-QML methods, respectively. The shaded area represents the standard error.}
    \label{fig:AccuracyStellar}
\end{figure}

\begin{figure}[htbp]
    \centering
    \includegraphics[width=0.9\linewidth]{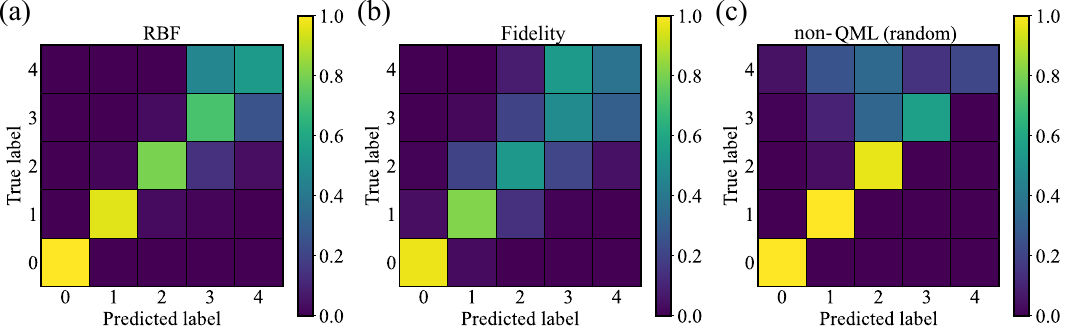}
    \caption{\textbf{Numerical results of confusion matrices for higher-stellar-rank classification.} The horizontal axis of the matrix represents the predicted label, while the vertical axis represents the true label. The color bar indicates the number of quantum data points in each grid, normalized over the true labels. Confusion matrices obtained from (a) the quantum RBF kernel, (b) the quantum fidelity kernel, and (c) the non-QML method.
}
    \label{fig:ConfusionMatrix}
\end{figure}

In this section, we perform numerical simulations of the classification of quantum datasets with higher stellar ranks to show that our quantum kernel method can also be applied to multiclass classification. Here, we consider single-mode quantum data represented by
\begin{equation}
    \ket{\psi} = \hat{D}(x_0 + \mathrm{i} p_0)\hat{R}(\theta_0)\hat{S}(r_0)\ket{n_0},
\end{equation}
where $x_0, p_0$ are sampled uniformly from $-1$ to $1$, $\theta_0$ is sampled uniformly from $0$ to $2\pi$, $r_0$ is sampled uniformly from 2 to 6 dB of squeezing, and $n_0\in\mathbb{Z}_{\geq0}$ is sampled uniformly from $0$ to $4$. Thus, our task is to classify the data into five labels (stellar rank) associated with $n_0$. We simulate this classification task using our QML method and the non-QML method. For the non-QML method, the classification follows the threshold values given in Ref.~\cite{PRXQuantum.2.020333}.

Figure~\ref{fig:AccuracyStellar} shows the classification accuracy as a function of the number of training quantum data, which does not consider experimental imperfections. In Fig.~\ref{fig:AccuracyStellar}, all of the methods we employ here successfully classify the quantum data, at least in the sense that the accuracies are higher than that of random guessing. The quantum RBF kernel achieves higher accuracy than the quantum fidelity kernel, showing the expressivity of the quantum RBF kernel. The difference between the methods is also evident in Fig.~\ref{fig:ConfusionMatrix}. Figure~\ref{fig:ConfusionMatrix} presents the confusion matrices for the quantum RBF kernel (Fig.~\ref{fig:ConfusionMatrix}(a)), the quantum fidelity kernel (Fig.~\ref{fig:ConfusionMatrix}(b)), and the non-QML method (Fig.~\ref{fig:ConfusionMatrix}(c)). Although the quantum kernel method can sometimes overestimate the stellar rank, it generally provides an ``unbiased'' label estimation than the non-QML approach. It can be understood as follows: the non-QML method relies on a sufficient condition for identifying the stellar rank. Hence, it may underestimate the true stellar rank but can never overestimate it, leading to a biased estimator. In contrast, the quantum kernel method is a QML approach that does not rely on such a sufficient condition, enabling unbiased classification.

\section{WORKING PRINCIPLE OF OUR PHOTONIC PROCESSOR} \label{sppl:WorkingPrincipleOf}

\begin{figure}[htbp]
    \centering
    \includegraphics[width=0.9\linewidth]{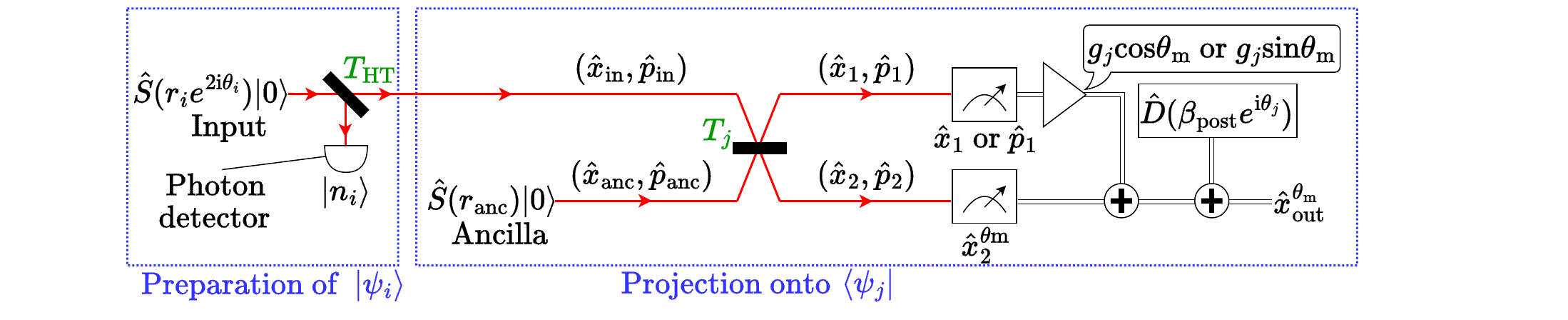}
    \caption{\textbf{Experimental setup to derive input-output relation.} We derive the working principle based on the notation described in this figure.}
    \label{fig:setup_sppl}
\end{figure}

Here, we first describe the quantum dataset in our experiment and then describe how to estimate the fidelities between the quantum data with our photonic processor. Our experiment focuses on the simplest nontrivial scenario, namely the classification of the single-mode quantum dataset of stellar rank 0 or 1, described by $\{\ket{\psi_i}=\hat{G}_{i}\ket{n_i}\}_{i=1}^N\ (n_i=0,1)$ and its label $\{z_i=n_i\}_{i=1}^N$. The arbitrary single-mode Gaussian operation $\hat{G}$ can be decomposed as $\hat{G}=\hat{D}(\beta)\hat{R}(\theta)\hat{S}(r)\hat{R}(\theta')$ with rotation $\hat{R}(\theta')$, squeezing $\hat{S}(r)$, another rotation $\hat{R}(\theta)$, and displacement $\hat{D}(\beta)$, where $\theta,\theta'\in\mathbb{R}$ and $r,\beta\in\mathbb{C}$~\cite{PhysRevA.71.055801, Fukui2022}. Hence, quantum data in our experiment can be described as
\begin{align}
    \ket{\psi}&=\hat{D}(\beta)\hat{R}(\theta)\hat{S}(r)\hat{R}(\theta')\ket{n}\nonumber\\
    &=\hat{D}(\beta)\hat{R}(\theta)\hat{S}(r)\ket{n}\quad (n=0,1).
    \label{eq:quantumdata}
\end{align}
In our experiment, we consider quantum data in Eq.~\eqref{eq:quantumdata} with $\beta=x_0+\mathrm{i}y_0\quad(x_0, y_0\in\{-1,0,1\}),\ \theta\in\{0,\pi/2\},\ r \in \left\{ k\ln(10)/20 \ \middle| \ k = 2,3,4,5,6 \right\}
$, and $n\in\{0,1\}$. Hence, the whole number of quantum data is $3^2\times2\times5\times2=180$. In Fig.~\ref{fig:kernel_exp}, these quantum data are sorted as follows: The rows and columns of the kernel matrix are ordered identically. The data are grouped according to their class label, with stellar-rank-zero states placed first and stellar-rank-one states placed second. Within each class block, the data are ordered by $\theta=0,\pi/2$, then by $x_0=-1,0,1$, then by $y_0=-1,0,1$, and finally by $k=2,\ldots,6$, where $k$ varies fastest.

Hereafter, we explain how to estimate the fidelities between quantum data with our photonic processor. To implement the displacement operations in the post-processing manner, we first rewrite the fidelity of two quantum data $\ket{\psi_i}$ and $\ket{\psi_j}$ as
\begin{align}
    |\braket{\psi_j|\psi_i}|^2
    &= |\braket{n_j| \hat{S}^\dagger(r_j)\hat{R}^\dagger(\theta_j)\hat{D}^\dagger(\beta_j)\hat{D}(\beta_i)\hat{R}(\theta_i)\hat{S}(r_i) | n_i}|^2 \nonumber \\
    &= |\braket{n_j| \hat{D}(\beta_{\mathrm{post}}) \hat{S}^\dagger(r_j)\hat{R}^\dagger(\theta_j)\hat{R}(\theta_i)\hat{S}(r_i) | n_i}|^2 \nonumber \\
    &= |\braket{n_j| \hat{D}(\beta_{\mathrm{post}}e^{\mathrm{i}\theta_j}) \hat{S}^\dagger(r_je^{2\mathrm{i}\theta_j}) \hat{S}(r_ie^{2\mathrm{i}\theta_i}) | n_i}|^2,
    \label{eq:expkernel}
\end{align}
where $\beta_{\mathrm{post}} = (\beta_i - \beta_j) \exp(-\mathrm{i}\theta_j)\cosh(r_j) + (\beta_i - \beta_j)^\ast \exp(\mathrm{i}\theta_j)\sinh(r_j)$. To measure the fidelity in Eq.~\eqref{eq:expkernel}, we proceed as follows: first, we prepare the quantum data $\hat{S}(r_ie^{2\mathrm{i}\theta_i})\ket{n_i}$, which omits the displacement operation from Eq.~\eqref{eq:quantumdata}; then, we experimentally implement the squeezing operation $\hat{S}^\dagger(r_je^{2\mathrm{i}\theta_j})$; we implement the displacement operation $\hat{D}(\beta_{\mathrm{post}}e^{\mathrm{i}\theta_j})$ by classical post-processing after measurement; finally, we estimate the population of the $n_j$-photon component.

To realize this measurement, we design the photonic quantum circuit shown in Fig.~\ref{fig:setup_sppl}, which is slightly modified from the setup in Fig.~\ref{fig:setup}(b), enabling the displacement operation to be implemented classically. The operation of this circuit is validated in the Heisenberg picture below.

We begin by preparing the input squeezed vacuum state $\hat{S}(r_ie^{2\mathrm{i}\theta_i})\ket{0}$. By detecting a single photon with an on/off detector in Fig.~\ref{fig:setup_sppl}, we can also prepare the input squeezed single-photon state $\hat{S}(r_ie^{2\mathrm{i}\theta_i})\ket{1}$. After preparing the input state, we implement a measurement-induced squeezing operation~\cite{PhysRevA.71.042308} by using another ancillary squeezed vacuum state $\hat{S}(r_{\mathrm{anc}})\ket{0}$. The quadrature amplitudes of these input and ancillary states are denoted as $(\hat{x}_{\mathrm{in}},\hat{p}_{\mathrm{in}})$ and $(\hat{x}_{\mathrm{anc}},\hat{p}_{\mathrm{anc}})$, respectively. Let $(\hat{x}_1,\hat{p}_1)$ and $(\hat{x}_2,\hat{p}_2)$ represent the quadrature amplitudes of the beams after passing through a beam splitter with transmissivity $T_j$, defined by
\begin{align}
    \hat{x}_1 &= \sqrt{1-T_j}\,\hat{x}_{\mathrm{in}} + \sqrt{T_j}\,\hat{x}_{\mathrm{anc}}, \\
    \hat{p}_1 &= \sqrt{1-T_j}\,\hat{p}_{\mathrm{in}} + \sqrt{T_j}\,\hat{p}_{\mathrm{anc}}, \\
    \hat{x}_2 &= \sqrt{T_j}\,\hat{x}_{\mathrm{in}} - \sqrt{1-T_j}\,\hat{x}_{\mathrm{anc}}, \\
    \hat{p}_2 &= \sqrt{T_j}\,\hat{p}_{\mathrm{in}} - \sqrt{1-T_j}\,\hat{p}_{\mathrm{anc}}.
\end{align}

In our experiment, we apply the squeezing operation with two phases, $\theta_j = 0$ and $\pi/2$, and the experimental settings depend on $\theta_j$. For $\theta_j=0$, we have $r_je^{2\mathrm{i}\theta_j}=r_j > 0$. Then, we choose $r_{\mathrm{anc}} < 0$ ($p_{\mathrm{anc}}$-squeezing), set $T_j = \exp(-2r_j)$, and measure $\hat{x}_1$ using the upper homodyne detector in Fig.~\ref{fig:setup_sppl}. In the lower homodyne detector, we measure the quadrature amplitude $\hat{x}^{\theta_{\mathrm{m}}}_2 = \hat{x}_2\cos\theta_{\mathrm{m}} + \hat{p}_2\sin\theta_{\mathrm{m}}$ while changing the measurement angle $\theta_{\mathrm{m}}$. The measurement outcomes $\hat{x}_{1}$ and $\hat{x}^{\theta_{\mathrm{m}}}_2$ from the upper and lower homodyne detectors are obtained, and we add the scaled value $g_j\cos\theta_{\mathrm{m}}\,\hat{x}_{1}$ to $\hat{x}^{\theta_{\mathrm{m}}}_2$. Moreover, the displacement operation $\hat{D}(\beta_{\mathrm{post}}e^{\mathrm{i}\theta_j})=\hat{D}(\beta_{\mathrm{post}})$ is implemented by classically adding $\Re(\beta_{\mathrm{post}})\cos\theta_{\mathrm{m}} + \Im(\beta_{\mathrm{post}})\sin\theta_{\mathrm{m}}$ to the measurement outcome above. Then, the post-processed output $\hat{x}^{\theta_{\mathrm{m}}}_{\mathrm{out}}$ is given by
\begin{align}
    \hat{x}^{\theta_{\mathrm{m}}}_{\mathrm{out}}
    &= \hat{x}^{\theta_{\mathrm{m}}}_2 + g_j\cos\theta_{\mathrm{m}}\,\hat{x}_{1} + \left[\Re(\beta_{\mathrm{post}})\cos\theta_{\mathrm{m}} + \Im(\beta_{\mathrm{post}})\sin\theta_{\mathrm{m}}\right] \nonumber \\
    &= \hat{x}_2\cos\theta_{\mathrm{m}} + \hat{p}_2\sin\theta_{\mathrm{m}} + \sqrt{\frac{1-T_j}{T_j}} \cos\theta_{\mathrm{m}}\,\hat{x}_{1} + \left[\Re(\beta_{\mathrm{post}})\cos\theta_{\mathrm{m}} + \Im(\beta_{\mathrm{post}})\sin\theta_{\mathrm{m}}\right] \nonumber \\
    &= (\sqrt{T_j}\,\hat{x}_{\mathrm{in}} - \sqrt{1-T_j}\,\hat{x}_{\mathrm{anc}})\cos\theta_{\mathrm{m}} + (\sqrt{T_j}\,\hat{p}_{\mathrm{in}} - \sqrt{1-T_j}\,\hat{p}_{\mathrm{anc}})\sin\theta_{\mathrm{m}} + \sqrt{\frac{1-T_j}{T_j}}\cos\theta_{\mathrm{m}}(\sqrt{1-T_j}\,\hat{x}_{\mathrm{in}} + \sqrt{T_j}\,\hat{x}_{\mathrm{anc}}) \nonumber \\
    &\quad + \left[\Re(\beta_{\mathrm{post}})\cos\theta_{\mathrm{m}} + \Im(\beta_{\mathrm{post}})\sin\theta_{\mathrm{m}}\right] \nonumber \\
    &= \frac{1}{\sqrt{T_j}}\hat{x}_{\mathrm{in}}\cos\theta_{\mathrm{m}} + \sqrt{T_j}\hat{p}_{\mathrm{in}}\sin\theta_{\mathrm{m}} - \sqrt{1-T_j}\hat{p}_{\mathrm{anc}}\sin \theta_{\mathrm{m}} + \left[\Re(\beta_{\mathrm{post}})\cos\theta_{\mathrm{m}} + \Im(\beta_{\mathrm{post}})\sin\theta_{\mathrm{m}}\right] \nonumber \\
    &= e^{r_j}\hat{x}_{\mathrm{in}}\cos\theta_{\mathrm{m}} + e^{-r_j}\hat{p}_{\mathrm{in}}\sin\theta_{\mathrm{m}} - \sqrt{1-e^{-2r_j}}\hat{p}_{\mathrm{anc}}\sin \theta_{\mathrm{m}} + \left[\Re(\beta_{\mathrm{post}})\cos\theta_{\mathrm{m}} + \Im(\beta_{\mathrm{post}})\sin\theta_{\mathrm{m}}\right].
    \label{eq:heisemberg1}
\end{align}
In the limit $r_{\mathrm{anc}}\rightarrow-\infty$, the term $\sqrt{1-e^{-2r_j}}\hat{p}_{\mathrm{anc}}$ vanishes. $(e^{r_j}\hat{x}_{\mathrm{in}}, e^{-r_j}\hat{p}_{\mathrm{in}})$ are affected by the squeezing operator $\hat{S}^\dagger(r_je^{2\mathrm{i}\theta_j})=\hat{S}^\dagger(r_j)$, and the term $\Re(\beta_{\mathrm{post}})\cos\theta_{\mathrm{m}} + \Im(\beta_{\mathrm{post}})\sin\theta_{\mathrm{m}}$ corresponds to the effect of the displacement operator $\hat{D}(\beta_{\mathrm{post}})$. Therefore, the output in Eq.~\eqref{eq:heisemberg1} corresponds to the measurement output of the state $\hat{D}(\beta_{\mathrm{post}})\hat{S}^\dagger(r_je^{2\mathrm{i}\theta_j})\hat{S}(r_ie^{2\mathrm{i}\theta_i})\ket{n_i}$ in the basis $\hat{x}\cos\theta_{\mathrm{m}} + \hat{p}\sin\theta_{\mathrm{m}}$. We measure the quadrature distributions at various angles $\theta_{\mathrm{m}}$, average these distributions, and perform an appropriately weighted integration to obtain the $n_j$-photon component~\cite{PhysRevA.52.R924}. From these procedures, the setup in Fig.~\ref{fig:setup_sppl} and Eq.~\eqref{eq:heisemberg1} allow for estimation of the quantum kernel defined in Eq.~\eqref{eq:expkernel}.

For $\theta_j=\pi/2$, we have $r_je^{2\mathrm{i}\theta_j}=-r_j<0$. Then, we choose $r_{\mathrm{anc}} > 0$ ($x_{\mathrm{anc}}$-squeezing), set $T_j = \exp(-2r_j)$, and measure $\hat{p}_1$ using the upper homodyne detector. We use the feedforward gain $g_j \sin\theta_{\mathrm{m}}$ for adding the $\hat{p}_1$ outcome, and apply the displacement by adding  $\Re(\mathrm{i}\beta_{\mathrm{post}})\cos\theta_{\mathrm{m}} + \Im(\mathrm{i}\beta_{\mathrm{post}})\sin\theta_{\mathrm{m}}$ to the measurement result. With these procedures, the post-processed output $\hat{x}^{\theta_{\mathrm{m}}}_{\mathrm{out}}$ is similarly calculated as
\begin{align}
    \hat{x}^{\theta_{\mathrm{m}}}_{\mathrm{out}}
    &= \hat{x}^{\theta_{\mathrm{m}}}_2 + g_j\sin\theta_{\mathrm{m}}\,\hat{p}_{1} + \left[\Re(\mathrm{i}\beta_{\mathrm{post}})\cos\theta_{\mathrm{m}} + \Im(\mathrm{i}\beta_{\mathrm{post}})\sin\theta_{\mathrm{m}}\right] \nonumber \\
    &= e^{-r_j}\hat{x}_{\mathrm{in}}\cos\theta_{\mathrm{m}} + e^{r_j}\hat{p}_{\mathrm{in}}\sin\theta_{\mathrm{m}} - \sqrt{1-e^{-2r_j}}\hat{x}_{\mathrm{anc}}\cos\theta_{\mathrm{m}}  + \left[\Re(\mathrm{i}\beta_{\mathrm{post}})\cos\theta_{\mathrm{m}} + \Im(\mathrm{i}\beta_{\mathrm{post}})\sin\theta_{\mathrm{m}}\right].
    \label{eq:heisemberg1neg}
\end{align}
Therefore, the $n_j$-photon component, corresponding to the quantum kernel in Eq.~\eqref{eq:expkernel}, can also be estimated under this condition.

\section{EXPERIMENTAL SETTINGS AND PARAMETERS}\label{sppl:ExperimentalSettingsAnd}

\begin{figure*}[htbp]
    \centering
    \includegraphics[width=0.9\linewidth]{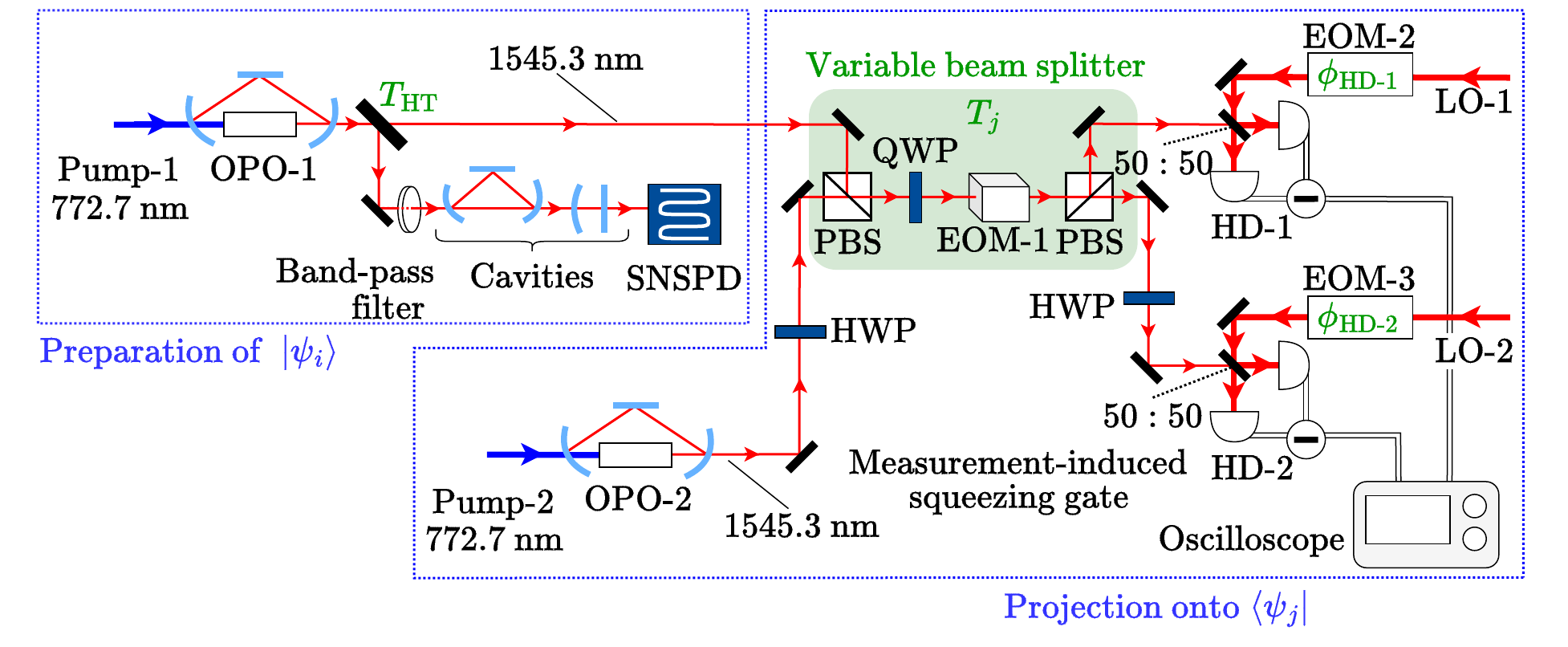}
    \caption{\textbf{Detailed optical setup used in our experiment.} See Supplementary Sec.~\ref{sppl:ExperimentalSettingsAnd} for details. EOM, electro-optic modulator; HD, homodyne detector; HWP, half-wave plate; LO, local oscillator; PBS, polarizing beam splitter; QWP, quarter-wave plate; OPO, optical parametric oscillator; SNSPD, superconducting nanostrip single-photon detector.}
    \label{fig:setup_methods}
\end{figure*}

Here, we describe the detailed parameters of the experimental setup in Fig.~\ref{fig:setup_methods}. Our experiment utilizes a continuous-wave laser operating at 1545.3 nm. Optical parametric oscillator-1 (OPO-1) produces the input states, whereas OPO-2 generates the ancillary states for the measurement-induced squeezing gate. These two OPOs, pumped by frequency-doubled beams at 772.7 nm, generate squeezed vacuum states. Both OPOs exhibit a spectral bandwidth (half-width at half maximum, HWHM) of 30 MHz. The pump power of OPO-1 is set to 10.1 mW, 22.4 mW, 39.1 mW, 59.7 mW, and 83.8 mW to achieve squeezing levels of approximately 2, 3, 4, 5, and 6 dB, respectively. OPO-2 is pumped at 200 mW to maximize its squeezing level. The squeezing and antisqueezing levels from OPO-2 are measured to be -4.4 dB and +9.1 dB. This measurement result indicates that the overall optical loss of the experimental setup is approximately 30\% and the pure squeezing level is 10 dB. These parameters are used in the theoretical model in Figs.~\ref{fig:kernel_exp}~and~\ref{fig:accuracy_exp} to take the experimental imperfection into account.

As an input state, we also create the squeezed single-photon state by subtracting a photon from the squeezed vacuum state of OPO-1. The photon subtraction is performed probabilistically using a superconducting nanostrip single-photon detector (SNSPD). The trigger light, which goes to the SNSPD, passes through a spectral filtering system before photon detection. This filtering restricts the frequency region of the detected photons and consequently limits that of the heralded squeezed single-photon states, making their bandwidth narrower than that of the ancillary squeezed state from OPO-2. Our filtering system has three parts: a dielectric multilayer band-pass filter with HWHM of $\sim 60$ GHz, a triangle cavity with HWHM of 11 MHz and free spectral range (FSR) of 4.5 GHz, and a Fabry-Perot cavity with HWHM of 144 MHz and FSR of 150 GHz. This combination ensures that only the trigger light within HWHM of 11 MHz around the center wavelength of the laser transmits through the filtering system. The detection rates are about 150, 350, 700, 1100, and 1700 counts per second (cps) when the squeezing levels $r_i$ correspond to 2, 3, 4, 5, and 6 dB, respectively. Fake and dark counts are measured to be approximately 20 cps and 3 cps, respectively.

The input and ancilla interfere at a variable beam splitter (VBS), realized by polarization optics consisting of a bulk electro-optic modulator (EOM-1), a quarter-wave plate (QWP), and polarizing beam splitters. The QWP is adjusted to achieve a transmissivity of 50\% when no voltage is applied to EOM-1, which makes it easy to lock the relative phase between the beams. Homodyne detectors measure quadrature amplitudes along homodyne angles $\phi_{\mathrm{HD-1}}$ and $\phi_{\mathrm{HD-2}}$. Their bandwidth is approximately 200 MHz, and the quantum efficiency of photodiodes in homodyne detection systems is designed to be 99\%. The VBS transmissivity ($T_j$) and the homodyne angles ($\phi_{\mathrm{HD-1}}, \phi_{\mathrm{HD-2}}$) are actively controlled by voltage-driven electro-optic modulators (EOM-1, EOM-2, and EOM-3) in Fig.~\ref{fig:setup_methods}.

Measurement outcomes from homodyne detection are acquired by an oscilloscope and analyzed using a classical computer. We measure 3000 time waveforms of quadratures at 12 phases in $15^\circ$ increments with HD-2 in Fig.~\ref{fig:setup_methods}. Simultaneously, HD-1 in Fig.~\ref{fig:setup_methods} measures 3000 time waveforms at a fixed phase ($\hat{x}_1$ or $\hat{p}_1$). Then, quadrature amplitudes of interest are extracted with a temporal mode function $f(t)$. The mode function is obtained by convolving the both-sided exponential correlation function $r(t)\propto\exp(-\alpha|t|)$ of the OPO squeezing with the time reversal of the response function of the filtering system $h(t)\propto\exp(-\beta t)u(t)$, where $u(t)$ is the unit step function~\cite{PhysRevA.96.052304}. Here, the parameters $\alpha$ and $\beta$ theoretically correspond to the bandwidth of OPO-1 and the filtering system, respectively. Then, the exact form of the normalized mode function is calculated as
\begin{equation}
    f(t) =\sqrt{\frac{\alpha\beta}{2\alpha+\beta}}\times
    \begin{cases}
     \frac{2\alpha}{\alpha-\beta}\exp(\beta t)-\frac{\alpha+\beta}{\alpha-\beta}\exp(\alpha t) & (t<0), \\
     \exp(-\alpha t) & (t\geq0).
    \end{cases}
\end{equation}
Experimentally, we optimize parameters $\alpha$ and $\beta$ to maximize the negativity of the Wigner function of squeezed single-photon states (quantum data with stellar rank 1), and the specific values used are $\alpha=2\pi\times39$ MHz and $\beta=2\pi\times10.8$ MHz.

% \bibliographystyle{apsrev4-2}
% \bibliography{supprefs} % Supplementary用Bibファイル

\end{document}